\documentclass{article}%
\usepackage{amsmath,amsthm}%
\usepackage{amsfonts, color}%
\usepackage{amssymb}%
\usepackage{graphicx}
\usepackage{verbatim} 
\usepackage[T1]{fontenc} 
\usepackage[utf8]{inputenc} 
\usepackage{authblk} 
\usepackage{booktabs}
\newtheorem{theorem}{Theorem}

\newtheorem{corollary}{Corollary}

\newtheorem{definition}{Definition}

\newtheorem{lemma}{Lemma}

\usepackage[margin=1in]{geometry}

\begin{document}

\title{Conditional Heteroskedasticity of Return Range Processes}
\author[1]{Yan Sun\thanks{Corresponding author. Email: yan.sun@usu.edu. Phone: (435)797-2861. Fax: (435)797-1822.}} 
\author[2]{Jennifer Loveland\thanks{jennifer.ellsworth@aggiemail.usu.edu}} 
\author[3]{Isaac Blackhurst\thanks{Isaac.Blackhurst@gmail.com}}
\affil[1,2,3]{Department of Mathematics $\&$ Statistics\\  
Utah State University\\
3900 Old Main Hill\\
Logan, Utah 84322-3900}

\date{}
\maketitle

\begin{abstract}
Price range contains important information about the asset volatility, and has long been considered an important indicator for it. In this paper, we propose to jointly model the [low, high] price range as a random interval and introduce an interval-valued GARCH (Int-GARCH) model for the corresponding [low, high] return range process. Model properties are presented under the general framework of random sets, and the parameters are estimated by a metric-based conditional least squares (CLS) method. Our empirical analysis of the daily return range data of Dow Jones component stocks yields very interesting results. 
\end{abstract}

%======================================================Introduction========================================================%
\section{Introduction}
%=======================================================================================================================%
Assets volatility plays an essential role in modern finance. It provides a measure of variability for the asset price over a certain period of time, and is a key parameter in many financial applications such as financial derivatives pricing, risk assessment, and portfolio management. Because volatility is not observable, statistical modeling that produces accurate volatility estimation and prediction is of great assistance. The squared return of log prices, being the classical estimator of variance, used to be the ``ideal'' proxy of volatility. As a result, many volatility models built on returns were proposed and, for a long time, had been very popular and successful. The celebrated ARCH (Engle 1982) and GARCH (Bollerslev 1986) models are examples of this type. Recently, as the high-frequency transaction data become available, price changes can be practically monitored in a continuous way, and the traditional low-frequency (e.g., daily) return is no longer quite representative of the volatility. For example, a small return does not necessarily imply low volatility, as the price may fluctuate a lot and close at a similar level to the opening. On the other hand, a big return could only be the result of a very different opening price from the previous day's closing. Therefore, the return-based models, using too little information, are likely to produce inefficient or even incorrect estimates of the volatility. 

In fact, since closing price is only a ``snapshot'' among numerous prices during a day, there is nothing special about it and return needs not to be calculated solely from it. With the availability of high-frequency data, intuitively, the log difference between any two observed prices in two consecutive days can be called a ``return''. This idea leads to a naturally generalized concept of daily return, which is an interval that includes all the ``snapshot'' returns. Let $y_t\left(s\right)$ be the log price of an asset at time $s$ on day $t$. Ideally $s$ should be a continuous time index. But since price can only be observed at discrete times even for high-frequency data, $s$ is assumed to be a discrete index here. We define the interval-valued daily return as
\begin{eqnarray}
  r_t&=&\left[\min_{s,w}\left\{y_t\left(s\right)-y_{t-1}\left(w\right)\right\}, \max_{s,w}\left\{y_t\left(s\right)-y_{t-1}\left(w\right)\right\}\right]\nonumber\\
  &=&\left[\min_{s}\left\{y_t\left(s\right)\right\}-\max_{w}\left\{y_{t-1}\left(w\right)\right\}, \max_{s}\left\{y_t\left(s\right)\right\}-\min_{w}\left\{y_{t-1}\left(w\right)\right\}\right].\label{def:rt-1}
\end{eqnarray}
Namely, $r_t$ is the range of ``snapshot'' returns during one day, which, apparently, contains more information about the daily volatility than the traditional closing-to-closing return. Our goal is to build a volatility model that reveals the dynamics of $r_t$ as an interval. 

Throughout the rest of the paper, we will view $r_t$ as a random interval and model its dynamics under the framework of random sets. To facilitate our presentation, we now briefly introduce the basic notations and definitions in the random set theory. (See, e.g., Kendall 1974, Matheron 1975, Artstein and Vitale 1975, Molchanov 2005, Sun and Ralescu 2014.) Let $(\Omega,\mathcal{L},P)$ be a probability space. Denote by $\mathcal{K}\left(\mathbb{R}^d\right)$ or $\mathcal{K}$ the collection of all non-empty compact subsets of $\mathbb{R}^d$. In the space $\mathcal{K}$, a linear structure is defined by Minkowski addition and scalar multiplication, i.e.,
\begin{equation}\label{set-lin}
  A+B=\left\{a+b: a\in A, b\in B\right\},\ \ \ \ \lambda A=\left\{\lambda a: a\in A\right\},
\end{equation}
$\forall A, B\in\mathcal{K}$ and $\lambda\in\mathbb{R}$. A natural metric for the linear space $\mathcal{K}$ is the Hausdorff metric $\rho_H$, which is defined as
\begin{equation*}
  \rho_H\left(A,B\right)=\max\left(\sup\limits_{a\in A}\rho\left(a,B\right), \sup\limits_{b\in B}\rho\left(b,A\right)\right),\ \forall A,B\in\mathcal{K},
\end{equation*}
where $\rho$ denotes the Euclidean metric.
A random compact set is a Borel measurable function $A: \Omega\rightarrow\mathcal{K}$, $\mathcal{K}$ being equipped with the Borel $\sigma$-algebra induced by the Hausdorff metric. For each $A\in\mathcal{K}\left(\mathbb{R}^d\right)$, the function defined on the unit sphere $S^{d-1}$:
\begin{equation*}
  s_A\left(u\right)=\sup_{a\in A}\left<u, a\right>,\ \ \forall u\in S^{d-1},
\end{equation*}
is called the support function of A. If $A(\omega)$ is convex almost surely, then $A$ is called a random compact convex set. Much of the random sets theory has focused on compact convex sets via their support functions. Especially, a one-dimensional random compact convex set is called a random interval.

Under the linear structure (\ref{set-lin}), the random interval $r_t$ can be alternatively defined as
\begin{equation}\label{def:rt-2}
r_t=\left[\min_{s}\left\{y_t\left(s\right)\right\}, \max_{s}\left\{y_t\left(s\right)\right\}\right]-\left[\min_{w}\left\{y_{t-1}\left(w\right)\right\}, \max_{w}\left\{y_{t-1}\left(w\right)\right\}\right].
\end{equation}
Namely, $r_t$ can be viewed as the return of daily price ranges. An immediate observation from equation (\ref{def:rt-2}) is
\begin{equation}\label{rt-prc-range}
  \left|r_t\right|=\left|\left[\min_{s}\left\{y_t\left(s\right)\right\}, \max_{s}\left\{y_t\left(s\right)\right\}\right]\right|+\left|\left[\min_{w}\left\{y_{t-1}\left(w\right)\right\}, \max_{w}\left\{y_{t-1}\left(w\right)\right\}\right]\right|,
\end{equation} 
where $\left|\cdot\right|$ denotes the Lebesgue measure or the length of an interval. That is, the length of $r_t$ is the sum of the high-low log price ranges of the two corresponding days. Feller (1951) derived the asymptotic distribution of the range of a Wiener process. This, together with the classical stochastic volatility models, provides the theoretical evidence that the log price range contains rich information of the integrated variance, and consequently the daily volatility too. It in turn explains from another perspective that $r_t$ is an invaluable resource for estimating the daily volatiity.

There has been a great deal of effort in the literature on volatility modeling using the high-low (log) price range of either low-frequency or high-frequency data (Garman and Klass 1980, Parkinson 1980, Rogers and Satchell 1991, Kunitomo 1992, Alizadel et al. 2002, Chou 2005, Engle and Gallo 2006, Brandt and Jones 2006, Christensen and Podolskij 2007). One common characteristic of the existing methods is that they consider the high-low price range as an indicator of the volatility and include it explicitly into the modeling by functions of certain forms. In this paper we take a different approach. We focus on the interval-valued [low, high] return range as opposed to the point-valued high-low price range. It is seen from the above discussion that the information from the high-low price range is already contained in the return range. In addition, the return range also includes the return itself somewhere in the interval. Therefore, careful modeling of the [low, high] return range as a whole is expected to produce promising estimate of the volatility that accounts for information from both the return and the price range. To this end, we propose to model the return range process $\left\{r_t\right\}$ by developing a generalized interval-valued GARCH (Int-GARCH) model. Under this new framework, the classical point-valued GARCH process becomes the degenerate case of our Int-GARCH model, where the return range $r_t$ has zero length. Our theoretical results are two-fold. We first show that under certain conditions our Int-GARCH model achieves a weak stationarity that is characterized by a time invariant mean and variance. Then, under the assumption of weak stationarity, we define and give the explicit formula of the autocorrelation function (ACF) of the Int-GARCH process. We propose a conditional least squares (CLS) method to estimate the model parameters, which is implemented by a Newton-Raphson algorithm. Simulation results are consistent with our theoretical findings and our CLS estimates are very satisfactory.

For empirical study, we analyze the Dow Jones component stocks data using our proposed Int-GARCH model. It is hard to judge how well it estimates the volatility, but we manage to demonstrate the characteristics and advantages of our Int-GARCH via comparisons to both GARCH and realized volatility (RV). It is shown that Int-GARCH model, by implicitly utilizing the intraday data (i.e., daily price range), has the advantage of RV to reflect the intraday price variability, which the return-based models such as GARCH are usually insufficient for. The idea to embed high-frequency volatility measures such as RV (Andersen et al. 2001, Barndorff-Nielsen and Shephard 2002) and realized kernel (Barndorff-Nielsen et al. 2008) into low-frequency models is not new. For example, the GARCHX (Engle 2002), HEAVY (Shephard and Sheppard 2010), and RealGARCH (Hansen et al. 2012) models are all important contributions along this direction. As certainly these compound models improve over the daily models with intraday information, their strong dependence on the realized measures makes them sensitive to any uncertainty or bias in these measures. For instance, it is well known that RV suffers from microstructure noises (see, e.g., Hansen and Lunde 2006, Brandi and Russell 2006, 2008). So, in order for such a compound model with RV to achieve its optimal performance, caution needs to be taken, either to choose the right sampling frequency or the appropriate remedy such as pre-averaging (Jacod et al. 2009) and subsampling (Zhang et al. 2004), to ensure that RV well measures the intraday volatility. Based on this discussion and the results of our empirical analysis, we believe that our Int-GARCH model makes the contribution of systematically integrating the low and high frequency volatility measures in an effective and operationally simple way, without obvious vulnerability to the microstructure noises. We present in our study that it is especially useful when assets are frequently traded and the intraday variation is significant.

The rest of the paper is organized as follows. We formally introduce our Int-GARCH model in Section 2. Main theoretical results are presented in Section 3 and Section 4. Section 5 proposes the conditional least squares method for estimating model parameters and carefully investigates its performances by a simulation study. Empirical study with the Dow Jones stocks data, as well as a detailed discussion, are reported in Section 6. We finish with concluding remarks in Section 7. Proofs and useful lemmas are deferred to the Appendix.

%===================================================Model specification=====================================================%
\section{The Int-GARCH model}
%=======================================================================================================================%
We assume observing an interval-valued time series $\left\{r_t\right\}_{t=1}^{T}$ of the form
$$r_t=[\lambda_t-\delta_t, \lambda_t+\delta_t],\ t=1,2,\cdots,T.$$
That is, $\left\{\lambda_t\right\}_{t=1}^{T}$ and $\left\{\delta_t\right\}_{t=1}^{T}$ are the associated center and radius processes, both of which are observable. A practical example of $\left\{r_t\right\}_{t=1}^{T}$ we consider in this paper is the daily return range process. Let $\mathcal{F}_t$ denote the information set up to time t, i.e.
$$\mathcal{F}_t=\sigma\left\{r_s: s\leq t\right\}.$$
We are concerned with the conditional variance of $r_t$ given $\mathcal{F}_{t-1}$.\\
\indent The GARCH model depicts the conditional variance of a point-valued return process as a linear function of the past returns and variances. This was inspired by the fact that assets returns usually exhibit volatility clustering: large variations in prices tend to cluster together, resulting in separate dynamic and tranquil periods of the market. Extending this spirit to the interval-valued process $\left\{r_t\right\}$, one would expect, conceptually, a model such like
\begin{equation}\label{model}
  H_t^2=g\left(H_s^2, r_s^2: s\leq t-1\right),
\end{equation} 
where $H_t^2$ denotes the conditional variance of $r_t$ given $\mathcal{F}_{t-1}$, and $g$ is linear in $H_s^2$ and $r_s^2$. We point out that $r_s^2$ in (\ref{model}) is only a notation, representing an estimate of the past variance. Actually, the square of a set has not been formally defined in the literature yet. To realize such a model, we need to have a mathematical definition for the conditional variance of a random interval $r_t$ and explicitly express it in terms of the observable random variables 
$\lambda_t$ and $\delta_t$.

The variance of a compact convex random set was originally introduced by Lyashenko (1982) and further studied in K\"orner (1995). Applying these results to the random interval $r_t$, via straightforward calculations, we obtain
\begin{eqnarray}
  \text{E}(r_t)&=&\left[\text{E}(\lambda_t)-\text{E}(\delta_t), \text{E}(\lambda_t)+\text{E}(\delta_t)\right],\label{meanA}\\
  \text{Var}(r_t)&=&\text{Var}(\lambda_t)+\text{Var}(\delta_t)\label{varA}.
\end{eqnarray}
The interpretation of (\ref{varA}) is obvious: the variance of a random interval consists of the variances from both the center and the radius. Considering (\ref{varA}) and assuming $\text{E}\lambda_t=0$, a reasonable function $g$ in (\ref{model}) seems to imply
\begin{equation*}
  H_t^2=\mu+\sum_{i=1}^{p}\alpha_i\left[\lambda_{t-i}^2+\delta_{t-i}^2-\left(E\delta_{t-i}\right)^2\right]
  +\sum_{i=1}^{q}\beta_iH_{t-i}^2,
\end{equation*}
where $p>0, q\geq 0$. Since the unconditional mean of $\delta_t$ is a constant and therefore can be absorbed into the parameter $\mu$, the above equation is simplified to
\begin{equation*}
  H_t^2=\mu+\sum_{i=1}^{p}\alpha_i\left[\lambda_{t-i}^2+\delta_{t-i}^2\right]+\sum_{i=1}^{q}\beta_iH_{t-i}^2.
\end{equation*}
To give more flexibility to our model, we allow for different degrees of dependence of $H_t$ on the past centers and radii. In addition, we propose to model volatility $H_t$ directly, instead of via the conditional variance $H_t^2$. 

Given the above discussion, we propose the following Int-GARCH $(p,q,w)$ model for the return range process:
\begin{eqnarray}
  r_t&=&h_t\cdot v_t,\label{igarch_1}\\
  v_t&=&[\epsilon_t-\eta_t, \epsilon_t+\eta_t],\label{igarch_2}\\
  \epsilon_t&\sim&N(0,1),\label{igarch_3}\\
  \eta_t&\sim&\Gamma(k,1),\label{igarch_4}\\
  h_t&=&\mu+\sum_{i=1}^{p}\alpha_i|\lambda_{t-i}|+\sum_{i=1}^{q}\beta_i\delta_{t-i}+\sum_{i=1}^{w}\gamma_ih_{t-i},\label{igarch_5}
\end{eqnarray}
where $p>0, q>0, w\geq0$, and $\left\{\alpha_i: i=1,\cdots,p\right\}, \left\{\beta_i: i=1,\cdots q\right\}, \left\{\gamma_i: i=1,\cdots,w\right\}$ are positive constants. In (\ref{igarch_1}), ``$\cdot$'' denotes the scalar multiplication. Although we pose parametric assumptions on the error terms $\epsilon_t$ and $\eta_t$ to simplify our presentation here, they are not really necessary. In practice, it is best to use the true data generating distributions, which vary from data to data. So, in replacement of (\ref{igarch_3})-(\ref{igarch_4}), a relaxed yet sufficient specification for $\epsilon_t$ and $\eta_t$ is
\begin{eqnarray}
  &&\text{E}\left(\epsilon_t|\mathcal{F}_{t-1}^{\epsilon, \eta}\right)=0,\label{igarch_3*}\\
  &&\text{E}\left(\eta_t|\mathcal{F}_{t-1}^{\epsilon, \eta}\right)=k,\label{igarch_4*}\\
  &&\eta_t>0\label{igarch_5*}.
\end{eqnarray}
We can get some insights into this interval-valued model by breaking it down into point-valued models. For example, let $hl_t$ denote the high-low log price range on day $t$. Then, in view of (\ref{rt-prc-range}), it is easily derived from the Int-GARCH model that 
\begin{equation*}
  hl_t+hl_{t-1}=2h_t\eta_t,
\end{equation*}
or equivalently,
\begin{equation}\label{range-model}
  hl_t=2h_t\eta_t-hl_{t-1}.
\end{equation}
Separately, denoting by $w_t$ the close-to-close return on day $t$, we have
\begin{equation}\label{return-model}
  w_t=h_t\left(\epsilon_t+a\eta_t\right),
\end{equation}
for some $a\in\left[-1, 1\right]$. The reduced model (\ref{range-model})-(\ref{return-model}) implies
\begin{eqnarray*}
  &&\text{Var}\left(hl_t|\mathcal{F}_{t-1}\right)=4kh_t^2,\\
  &&\text{Var}\left(w_t|\mathcal{F}_{t-1}\right)=\left(1+a^2k\right)h_t^2.
\end{eqnarray*}
The return $w_t$ can be anywhere in the return range $r_t$, so the exact value of $a$ in the above equations is not known. But it should not be important anyway, as a single close-to-close return is not necessarily informative about the volatility. The idea of Int-GARCH is to combine all these factors and calculate the conditional variance of the whole interval $r_t$ as 
\begin{equation*}
  H_t^2=\text{Var}(h_t\epsilon_t)+\text{Var}(h_t\eta_t)=h_t^2(k+1)\propto h_t^2,
\end{equation*}
where $H_t$ is used as the estimated volatility. Finally, we notice that when $k\to 0$, the radius $\delta_t\to 0$, and the model (\ref{igarch_1})-(\ref{igarch_5}) reduces to the usual point-valued GARCH model, except that the linear dependence on the past is specified for the conditional standard deviation.

%===============================================Distribution Int-GARCH(1,1,1)======================================================%
\section{Distribution of Int-GARCH (1,1,1)}
%===========================================================================================================================%
Similar to the GARCH model, the Int-GARCH(1,1,1) process is the simplest but most often an effective model for analyzing interval-valued time series with conditional heteroskedasticity. In this section, we derive several important distribution properties of Int-GARCH (1,1,1). Before we present our theoretical results, we first notice that for the Int-GARCH (1,1,1) process,
\begin{eqnarray*}
  h_t&=&\mu+\alpha_1|\lambda_{t-1}|+\beta\delta_{t-1}+\gamma_1h_{t-1}\\
  &=&\mu+\alpha_1|\epsilon_{t-1}|h_{t-1}+\beta\eta_{t-1}h_{t-1}+\gamma_1h_{t-1}\\
  &=&\mu+\left(\alpha_1|\epsilon_{t-1}|+\beta\eta_{t-1}+\gamma_1\right)h_{t-1}.
\end{eqnarray*}
Therefore, defining the i.i.d. random variables $x_t=\alpha_1|\epsilon_{t-1}|+\beta\eta_{t-1}+\gamma_1$, $t\in\mathbb{N}$, $h_t$ can be re-written as
\begin{equation}\label{def_x}
  h_t=\mu+x_th_{t-1}.
\end{equation}
We will use (\ref{def_x}) throughout this section to simplify notations.

%===========================================================================================================================%
 \subsection{Weak stationarity}
It is derived in K\"orner (1995) that the covariance between two random intervals is the sum of the covariances between the two centers and two radii. This implies
\begin{equation}\label{covAB}
  \text{Cov}(r_t, r_s)=\text{Cov}(\lambda_t, \lambda_s)+\text{Cov}(\delta_t, \delta_s),\ s,t\in\mathbb{N}.
\end{equation}
We are ready to extend the notion of weak stationarity to interval-valued time series in the obvious way. 
\begin{definition}
An interval-valued time series $\left\{r_t\right\}$ is said to be weakly stationary, or second-moment stationary, if its unconditional mean $\text{E}\left(r_t\right)$ and covariance 
$\text{Cov}\left(r_t, r_{t+s}\right)$ exist and are independent of time $t$ for all integers $s$, where $\text{E}\left(r_t\right)$ and $\text{Cov}\left(r_t, r_{t+s}\right)$ are given in (\ref{meanA}) and (\ref{covAB}), respectively.
\end{definition}
The existences of $\text{E}\left(r_t\right)$ and $\text{Var}\left(r_t\right)$ are closely related to the those of $\text{E}\left(h_t\right)$ and $\text{E}\left(h_t^2\right)$, respectively. In fact, $Eh_t^2<\infty$ implies the existences of the first two moments of $r_t$. We give precise conditions in the following two theorems.
%=======================================================Theorem 1===========================================================%
\begin{theorem}\label{thm:mean} 
Consider the Int-GARCH model (\ref{igarch_1})-(\ref{igarch_5}) with $p=q=w=1$. Assume $\left\{r_t\right\}$ starts from its infinite past with a finite mean. Then, $Eh_t<\infty$
if and only if $Ex_t<\infty$, i.e.
$$\alpha_1\sqrt{\frac{2}{\pi}}+\beta_1k+\gamma_1<1.$$
When this condition is satisfied, 
\begin{equation}\label{mean_h}
 Eh_{t}=\frac{\mu}{1-\alpha_{1}\sqrt{2/\pi}-\beta_{1}k-\gamma_{1}},
\end{equation}
and
\begin{equation}\label{mean_e}
 Er_{t}=\left[-kE\left(h_{t}\right), kE\left(h_{t}\right)\right].
\end{equation}
 \end{theorem}
 
%=======================================================Theorem 2===========================================================%   
\begin{theorem}\label{thm:var}
 Consider the Int-GARCH(1,1,1) model $\left\{r_t\right\}$ as in Theorem \ref{thm:mean}. Assume $\left\{r_t\right\}$ starts from its infinite past with a finite variance. Then  
 $E\left(h_{t}^{2}\right)<\infty$ if and only if $E\left(x_{t}^{2}\right)<E\left(x_{t}\right)<1$, i.e.
 \begin{eqnarray*}
 && (i)\ \alpha_1\sqrt{\frac{2}{\pi}}+\beta_1k+\gamma_1<1;\\
 && (ii)\ \frac{\alpha_{1}^{2}+\beta_{1}^{2}\left(k+k^{2}\right)+\gamma_{1}^{2}+2\alpha_{1}\beta_{1}\sqrt{\dfrac{2}{\pi}}k+2\alpha_{1}\gamma_{1}\sqrt{\dfrac{2}{\pi}}+2\beta_{1}\gamma_{1}k}
 {\alpha_1\sqrt{\frac{2}{\pi}}+\beta_1k+\gamma_1}<1.
 \end{eqnarray*}
 When these conditions are satisfied,
 \begin{eqnarray}\label{mean_h2}
   E\left(h_{t}^{2}\right)=\mu^{2}\dfrac{C_{1}+1}{\left(C_{2}-1\right)\left(C_{1}-1\right)},
 \end{eqnarray}
 and
 \begin{eqnarray}\label{var_e}
   Var\left(r_{t}\right)=\left(1+k+k^{2}\right)E\left(h_{t}^{2}\right)-k^{2}\left[E\left(h_{t}\right)\right]^{2},
 \end{eqnarray}
 where $E\left(h_t\right)$ is given in (\ref{mean_h}), and $C_{1}=E\left(x_{t}\right)$, $C_{2}=E\left(x_{t}^{2}\right)$.
\end{theorem} 
So far we have found equivalent conditions for the existence of a time-invariant (unconditional) mean and variance. In order to guarantee weak stationarity, by definition, we still need to find conditions under which the (unconditional) covariances are finite and time-invariant. According to the following Theorem \ref{thm:cov}, these conditions turn out to be the same as those for the existence of variance. This is not surprising as 
\begin{eqnarray*}
  |\text{Cov}\left(r_t, r_{t+h}\right)| \leq
  |\text{Cov}\left(\lambda_t, \lambda_{t+h}\right)|+|\text{Cov}\left(\delta_t, \delta_{t+h}\right)|
  %&\leq& \sqrt{\text{Var}\left(\lambda_t\right)\text{Var}\left(\lambda_{t+h}\right)}+\sqrt{\text{Var}\left(\delta_t\right)\text{Var}\left(\delta_{t+h}\right)}\\
  \leq \text{Var}\left(\lambda_t\right)+\text{Var}\left(\delta_t\right),
\end{eqnarray*}
and, for a time series model with specified recursive structure, the existence of any unconditional moment usually implies time-invariance too. We summarize this conclusion in Corollary \ref{cor:stat} following Theorem \ref{thm:cov}. 

%=======================================================Theorem 3===========================================================% 
 \begin{theorem}\label{thm:cov} 
Consider the Int-GARCH(1,1,1) process $\left\{r_t\right\}$. Under the assumptions of Theorem \ref{thm:var}, the covariance of any two random intervals $r_t$ and $r_{t+s}$ is given by 
\begin{align*}
\mbox{Cov}\left(r_{t},r_{t+s}\right) & =\begin{cases}
\left(1+k+k^{2}\right)E\left(h_{t}^{2}\right)-k^{2}\left[E\left(h_{t}\right)\right]^{2}, & s=0;\\
kE\left(h_{t}h_{t+s}\eta_{t}\right)-k^{2}\left[E\left(h_{t}\right)\right]^{2}, & |s|>0,
\end{cases}
\end{align*}
where $E\left(h_t\right)$ and $E\left(h_t^2\right)$ are given in (\ref{mean_h}) and (\ref{mean_h2}), respectively, and $E\left(h_{t}h_{t+s}\eta_{t}\right)$ is calculated explicitly in Lemma \ref{eta-X and h-h-eta} (see Appendix).
\end{theorem}

%=======================================================Corollary 1===========================================================%  
\begin{corollary}\label{cor:stat} 
The Int-GARCH(1,1,1) process is weakly stationary, or second-moment stationary, if and only if $E\left(x_{t}^{2}\right)<E\left(x_{t}\right)<1$.
\end{corollary}

%===========================================================================================================================%
 \subsection{Auto-correlation function (ACF)}
 The notion of the variance and covariance for compact convex random sets were naturally extended to the correlation coefficient of two random sets $A$ and $B$, which is defined as
 \begin{equation}\label{corr}
   \text{Corr}\left(A,B\right)=\frac{\text{Cov}\left(A,B\right)}{\sqrt{\text{Var}\left(A\right)\text{Var}\left(B\right)}}.
 \end{equation} 
Based on this definition, we calculate the auto-correlation function (ACF) of the Int-GARCH(1,1,1) process and give the result in the corollary below. 
 %=======================================================Corollary 2===========================================================%
\begin{corollary}\label{cor:acf} 
Under the assumptions of Theorem \ref{thm:var}, the auto-correlation function of the Int-GARCH(1,1,1) process
$\left\{ r_{t}\right\} $ is 
\[
\rho(s)=\begin{cases}
1, & s=0\\
\dfrac{kE\left(h_{t}h_{t+s}\eta_{t}\right)-k^{2}\left[E\left(h_{t}\right)\right]^{2}}{\left(1+k+k^{2}\right)E\left(h_{t}^{2}\right)-k^{2}\left[E\left(h_{t}\right)\right]^{2}}, & |s|>0.
\end{cases}
\]
 \end{corollary}
 
We plot the ACF for a specific Int-GARCH(1,1,1) model (Model I in the simulation) in Figure 1. We see that the centers are uncorrelated. This has been verified by (\ref{eqn2}) in the proof of Theorem \ref{thm:cov}. (We will elaborate more on this in the subsequent section.) The radii, or the lengths of the intervals, have a relatively persistent auto-correlation, which coincides with the phenomenon of ``volatility clustering''. This long-term dependence of radii carries over to the intervals as a whole, and results in a slow-dying ACF of the interval-valued process. \\

%==========================================================Figure 1=============================================================%
\begin{figure}[ht]
\centering
\includegraphics[ height=1.800in, width=2.300in]{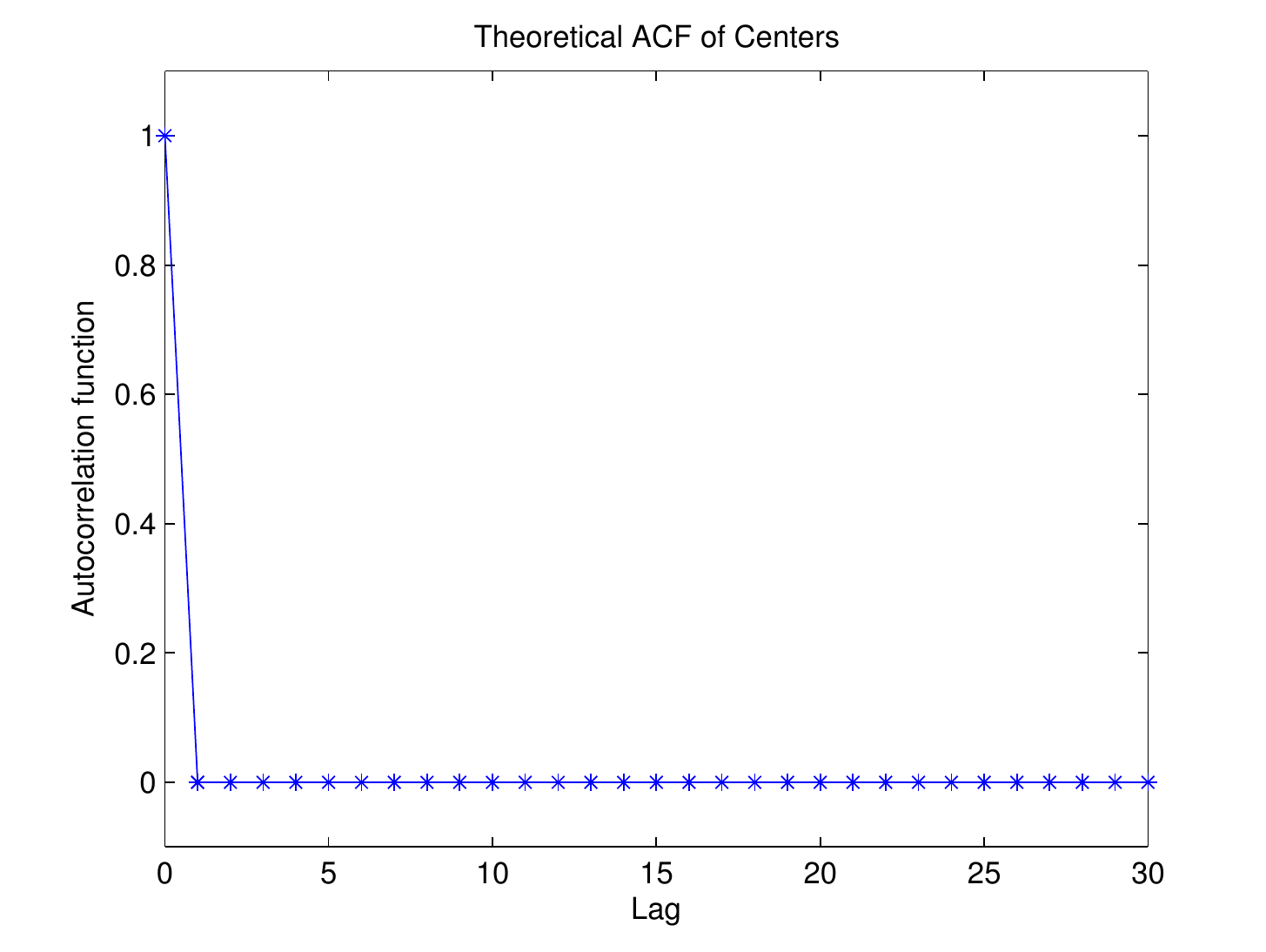}
\includegraphics[ height=1.800in, width=2.300in]{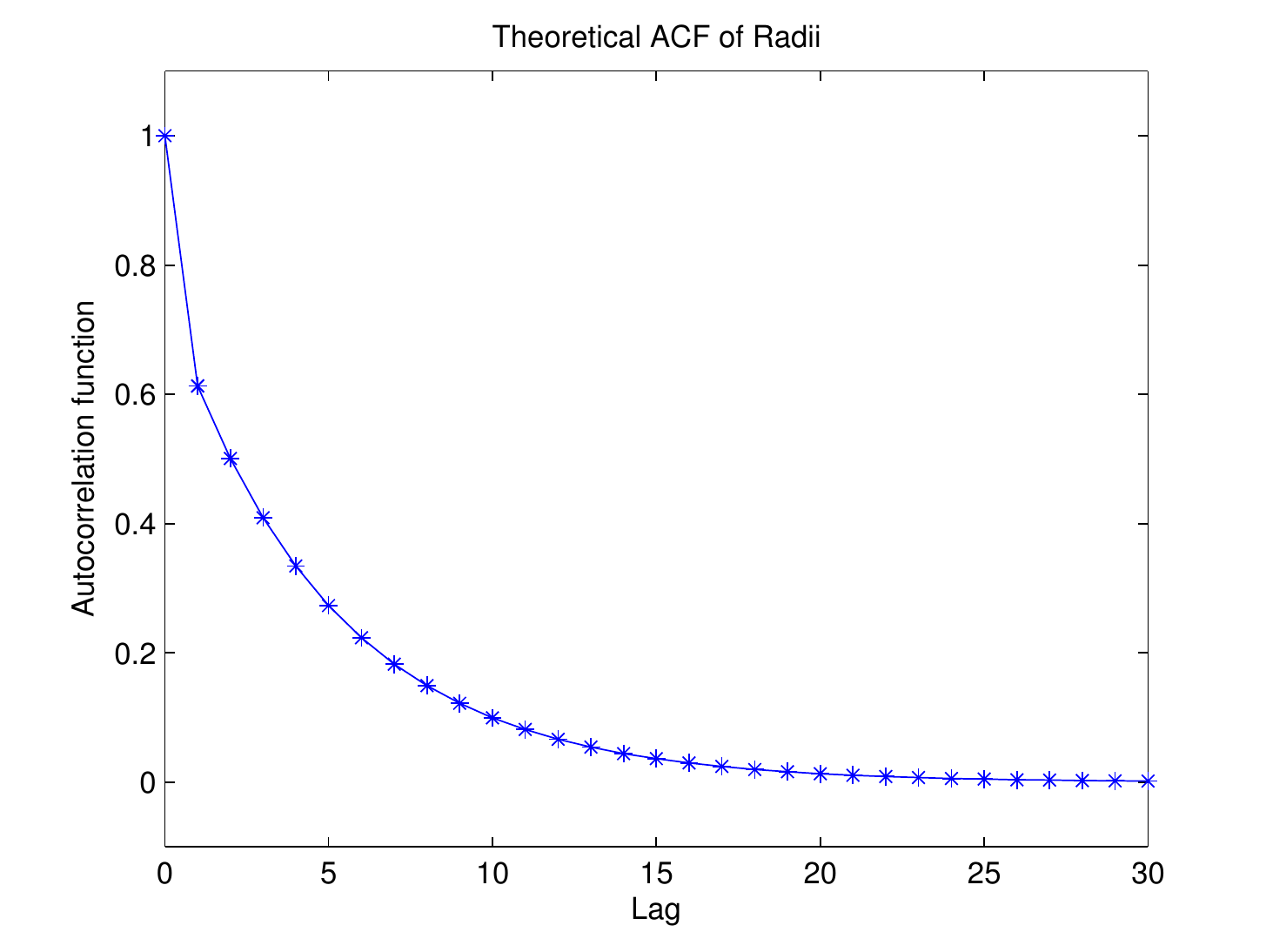}\\
\includegraphics[ height=1.800in, width=2.300in]{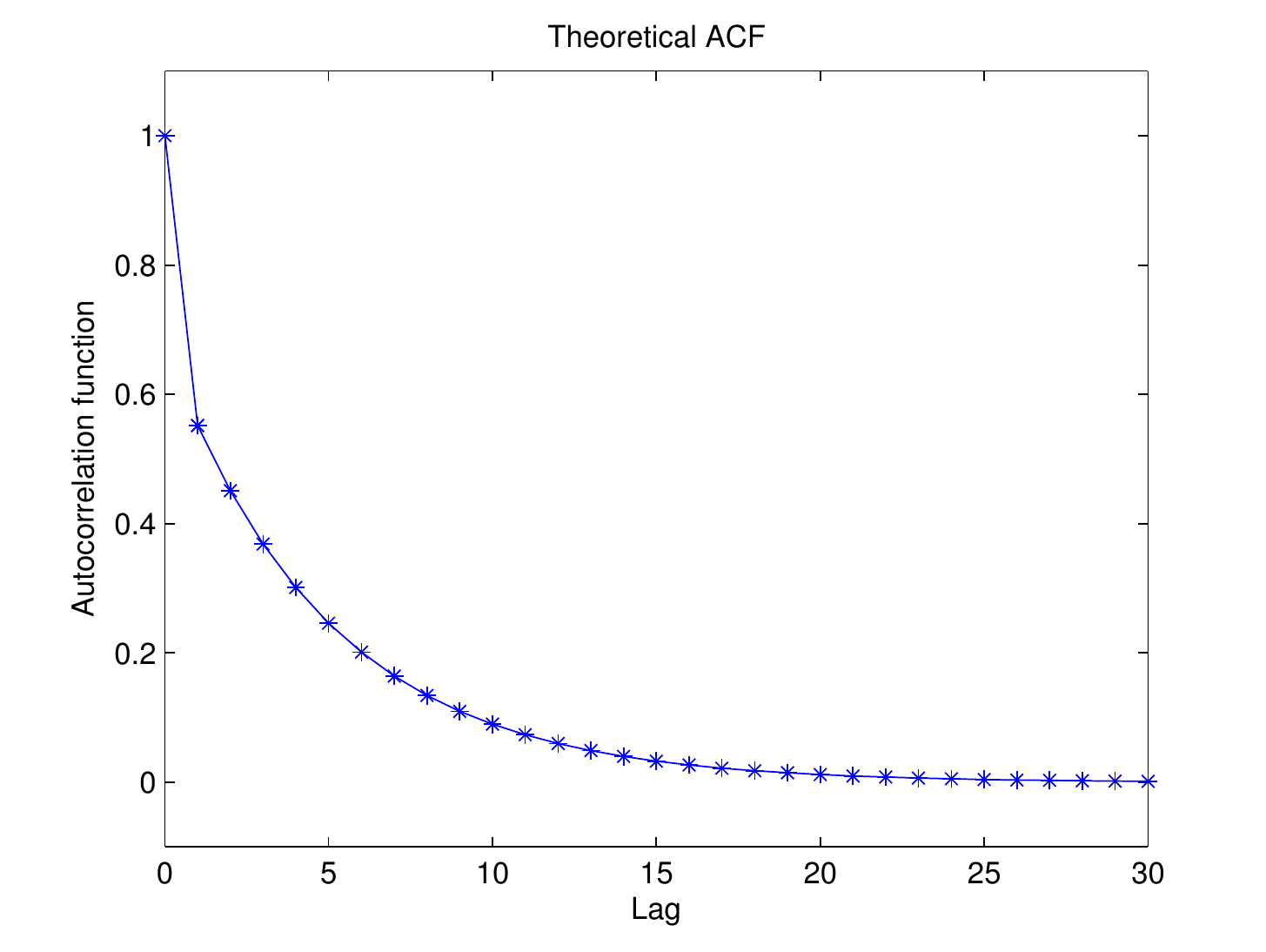}
\caption{Theoretical auto-correlation function of Model I.}
\label{fig:realACF_1}
\end{figure}

%==============================================================================================================================%
%===================================================Int-GARCH (p,q,w,)============================================================%
\section{The general Int-GARCH (p,q,w)}
%==============================================================================================================================% 
\subsection{Mean stationarity}
We provide the necessary and sufficient conditions of mean stationarity for the general Int-GARCH (p,q,w) model in the following theorem.
\begin{theorem}\label{thm:mean-gen} 
Consider the general Int-GARCH model (\ref{igarch_1})-(\ref{igarch_5}). Define
\begin{equation}
  x_{i,t}=\alpha_i|\epsilon_t|I_{\left\{1\leq i\leq p\right\}}+\beta_i|\eta_t|I_{\left\{1\leq i\leq q\right\}}+\gamma_iI_{\left\{1\leq i\leq w\right\}},\label{def:X_it}
\end{equation}
and
\begin{equation}
  E\left(X_{i,t}\right)=\mu_i,\label{def:mu_i}
\end{equation}
where $i=1,2,\cdots,k=\max\left\{p,q,w\right\}$. Assume $\left\{r_t\right\}$ starts from its infinite past with a finite mean. Then, $Eh_t<\infty$
if and only if $\sum_{i=1}^{k}\mu_i<1$. When this condition is satisfied, 
\begin{equation}\label{mean_h}
 Eh_{t}=\frac{\mu}{1-\sum_{i=1}^{k}\mu_i},
\end{equation}
and
\begin{equation}\label{mean_e}
 Er_t=\left[-kE\left(h_{t}\right), kE\left(h_{t}\right)\right].
\end{equation}
 \end{theorem}

%===========================================================================================================================% 
 \subsection{Relationship to GARCH}
 It is well known that the ARCH/GARCH types of models are heteroskedastic but serially uncorrelated. From the proof of Theorem \ref{thm:cov}, the centers in the Int-GARCH(1,1,1) model are also uncorrelated. And, in general, the centers in any Int-GARCH(p,q,w) model are uncorrelated. This is not a coincidence. By noticing 
 $\lambda_t=\epsilon_th_t$, we see that $h_t$ is the conditional standard deviation of the center $\lambda_t$. The GARCH(p,q) process models $h_t$ by
 \begin{equation}
   h_t^2=\mu+\sum_{i=1}^{p}\alpha_i\lambda_{t-i}^2+\sum_{i=0}^{q}\beta_ih_{t-i}^2,\ p>0,\ q\geq 0. 
 \end{equation}
Compared to this, except for using $h_t$ directly instead of $h_t^2$, our Int-GARCH(p,q,w) specification for $h_t$ (\ref{igarch_5}) has only one extra term, which is a linear combination of the past radii $\left\{\delta_{t-i}: i=1,\cdots,q\right\}$. Therefore, the center of our Int-GARCH model is just a GARCH process with one external variable that is the half-range of the price. Our Int-GARCH process as a whole further systematically models the interaction between the average and the range of prices under the random sets 
framework.

%==============================================================================================================================%
%========================================================CLSE==================================================================%
\section{The conditional least squares (CLS) estimate}\label{CLSE}
We propose a conditional least squares method to estimate the parameters of the Int-GARCH model. Throughout this section, we denote by $\vec{\theta}$ the parameter vector that contains all the Int-GARCH parameters $\left\{k; \mu; \alpha_1, \cdots, \alpha_p; \beta_1, \cdots, \beta_q; \gamma_1, \cdots, \gamma_w\right\}$ as its components. Conditioning on the past, the expected value of the current interval $r_t$ is
\begin{eqnarray*}
\hat{r_{t}} & = & E\left[r_{t}\mid\mathcal{F}_{t-1}\right]\\
 & = & E\left\{\left[h_{t}\left(\varepsilon_{t}-\eta_{t}\right),h_{t}\left(\varepsilon_{t}+\eta_{t}\right)\right]|\mathcal{F}_{t-1}\right\}\\
 & = & \left[E\left\{h_t\left(\varepsilon_{t}-\eta_{t}\right)|\mathcal{F}_{t-1}\right\}, E\left\{h_t\left(\varepsilon_{t}+\eta_{t}\right)|\mathcal{F}_{t-1}\right\}\right]\\
 & = & \left[-kh_{t},kh_{t}\right],
\end{eqnarray*}
where
\begin{equation}\label{pred_ht}
   h_t=\mu+\sum_{i=1}^{p}\alpha_{i}\left|\lambda_{t-i}\right|+\sum_{i=1}^{q}\beta_{i}\delta_{t-i}+\sum_{i=0}^{w}\gamma_{i}h_{t-i}.
\end{equation}
A widely used metric in the space $\mathcal{K}_{C}$ of compact convex subsets of $\mathbb{R}^d$ is given by
\begin{equation}\label{delta-metric}
\delta\left(A, B\right)=\left[d\int_{S^{d-1}}|s_A(u)-s_B(u)|^2\mu\left(\mathrm{d}u\right)\right]^{\frac{1}{2}},
\end{equation}
where $A, B\in\mathcal{K}_{C}$ and $\mu$ is the normalized Lebesgue measure on $S^{d-1}$ (see, e.g., K\"{o}rner 1995). Letting $A=r_t, B=\hat{r}_t$ in (\ref{delta-metric}), we have
\begin{equation*}
\delta\left(r_t, \hat{r}_t\right)_2^2
=\frac{1}{2}\left\{\left[\lambda_{t}-\delta_{t}+kh_t\right]^{2}+\left[\lambda_{t}+\delta_{t}-kh_t\right]^{2}\right\} ,
\end{equation*}
Our CLS estimate $\hat{\vec{\theta}}$ is defined such that the sum of the squared $\delta$-metric between $r_t$ and $\hat{r}_t$ is minimized, i.e.
\begin{equation}\label{clse}
  \hat{\vec{\theta}}=\arg\min_{\vec{\theta}}\left\{\sum_{t=1}^{T}\left\|r_t-\hat{r}_t\right\|_2^2\right\}:=\arg\min_{\vec{\theta}}\left\{L(\vec{\theta})\right\}.
\end{equation}
Notice that the recursive formula in (\ref{pred_ht}) requires starting values $h_0$ and $r_0$. We assume that the process $\left\{r_t\right\}$ starts from its infinite past with a finite mean and variance, and therefore it is reasonable to let $h_0=E\left(h_t\right)$. An alternative is to let $h_0=0$, assuming $\left\{r_t\right\}$ starts from a constant interval $r_0$. Based on our experience, the choice of $h_0$ does not affect the accuracy of the final estimate, given a reasonably large sample. In either case, we let $r_0=E\left(r_t\right)$.

In principle, an estimator that makes full use of the distribution information such as the maximum likelihood estimate (MLE) is more desired, especially in the situation with heteroskedasticity. However, the (conditional) likelihood for the Int-GARCH model involves distribution functions of both $\epsilon_t$ and $\eta_t$, and so the (conditional) MLE is computationally much more expensive than the CLS. As will be seen in the empirical analysis section, it is the interval length that plays the dominant role in the model. This implies that the distribution information of $\eta_t$, which is largely reflected on the value of $k$, is dominating that of $\epsilon_t$. By switching from conditional MLE to CLS, we lose the minor information of $\epsilon_t$, but our gain of computational efficiency is substantial. In the rest of the section, we will restrict our attention to the Int-GARCH(1,1,1) model. We first give explicit formulae for calculating the CLS estimate using a Newton-Raphson algorithm. This necessitates the computation of initial parameters. An initialization scheme is provided based on the method of moments. Finally, we carry out a simulation study to examine the performance of our proposed CLS method, and the results are very satisfactory. 

\subsection{A Newton-Raphson algorithm}
The parameter vector in the Int-GARCH(1,1,1) model is $\vec{\theta}=\left[k, \mu, \alpha_1, \beta_1, \gamma_1\right]^{T}$. Plugging (\ref{pred_ht}) in (\ref{clse}), the conditional least squares function $L\left(\vec{\theta}\right)$ is explicitly expressed as
\begin{align}
L\left(\vec{\theta}\right)= & \frac{1}{2}\sum_{t=1}^{T}\left\{ \left[\lambda_{t}-\delta_{t}+k\left(\mu+\alpha_{1}\left|\lambda_{t-1}\right|+\beta_{1}\delta_{t-1}+\gamma_{1}h_{t-1}\right)\right]^{2}\right.\nonumber\\
 & +\left.\left[\lambda_{t}+\delta_{t}-k\left(\mu+\alpha_{1}\left|\lambda_{t-1}\right|+\beta_{1}\delta_{t-1}+\gamma_{1}h_{t-1}\right)\right]^{2}\right\}. 
\end{align}
Consequently, the gradient vector $L^{\prime}\left(\vec{\theta}\right)$ and the Hessian matrix $L^{\prime\prime}\left(\vec{\theta}\right)$ are found to be
\begin{eqnarray*}
L^{\prime}\left(\vec{\theta}\right) & = & 2\sum_{t=1}^{T}\left\{ \left(-\delta_{t}+kh_{t-1}\right)\left[\begin{array}{c}
k\\
k\left|\lambda_{t-1}\right|\\
k\delta_{t-1}\\
kh_{t-1}
\end{array}\right]\right\} ,
\end{eqnarray*}
and 
\begin{eqnarray*}
L^{\prime\prime}\left(\vec{\theta}\right) & =2\sum_{t=1}^{T} & \left[\begin{array}{cccc}
k^{2} & k^{2}\left|\lambda_{t-1}\right| & k^{2}\delta_{t-1} & k^{2}h_{t-1}\\
k^{2}\left|\lambda_{t-1}\right| & k^{2}\left|\lambda_{t-1}\right|^{2} & k^{2}\delta_{t-1}\left|\lambda_{t-1}\right| & k^{2}h_{t-1}\left|\lambda_{t-1}\right|\\
k^{2}\delta_{t-1} & k^{2}\left|\lambda_{t-1}\right|\delta_{t-1} & k^{2}\delta_{t-1}^{2} & k^{2}h_{t-1}\delta_{t-1}\\
k^{2}h_{t-1} & k^{2}\left|\lambda_{t-1}\right|h_{t-1} & k^{2}\delta_{t-1}h_{t-1} & k^{2}h_{t-1}^{2}
\end{array}\right].
\end{eqnarray*}
The Newton-Raphson Algorithm thus consists of an iterative computation of the following formula 
\[
\vec{\theta}^{(k+1)}=\vec{\theta}^{(k)}-\left[L''\left(\vec{\theta}^{(k)}\right)\right]^{-1}\cdot L'\left(\vec{\theta}^{(k)}\right),
\]
where $L'\left(\vec{\theta}^{(k)}\right)$and $\left[L''\left(\vec{\theta}^{(k)}\right)\right]^{-1}$ are the gradient vector and inverse Hessian matrix evaluated at the $k^{th}$ step estimate $\vec{\theta}^{(k)}$, respectively.

\subsection{\label{sub:Initialization-of-the}Parameter initialization}
We use the method of moments to get an initial estimate of $\vec{\theta}$. Notice that
\begin{eqnarray}
  E\left(\delta_t\right) &=& E\left(h_t\delta_t\right)=E\left(h_t\right)E\left(\delta_t\right)=kE\left(h_t\right),\label{ini_1}\\
  E\left|\lambda_{t}\right| &=& E\left|h_{t}\varepsilon_{t}\right|=E\left(h_{t}\right)E\left|\varepsilon_{t}\right|=\sqrt{2/{\pi}}E\left(h_{t}\right),\label{ini_2}\\
  E\left(h_{t}\right) &=& \dfrac{\mu}{1-\alpha_{1}\sqrt{2/\pi}-\beta_{1}k-\gamma_{1}}.\label{ini_3}
\end{eqnarray}
Equating $E\left|\lambda_t\right|$ with its sample mean $\overline{\left|\lambda_t\right|}$ in (\ref{ini_2}), we get the moment estimate of $E(h_t)$. That is,
\begin{equation*}
  \overline{E\left(h_t\right)}=\sqrt{\frac{\pi}{2}}\overline{\left|\lambda_t\right|}.
\end{equation*}
Then, replacing $E\left(h_{t}\right)$ by $ \overline{E\left(h_t\right)}$ and equating $E\left(\delta_t\right)$ with its sample mean $\overline{\left|\delta_t\right|}$ in (\ref{ini_1}), we obtain an initial estimate of $k$. Similarly, $\mu$ is initialized by replacing $E\left(h_{t}\right)$ by $\overline{E\left(h_t\right)}$ and a rough guessing (for example, 0.4) of
$1-\alpha_{1}\sqrt{2/\pi}-\beta_{1}k-\gamma_{1}$ in (\ref{ini_3}). Finally, initial estimates of $\alpha_1$, $\beta_1$, and $\gamma_1$ are obtained by setting each of 
$\alpha_{1}\sqrt{2/\pi}$, $\beta_{1}k$, and $\gamma_{1}$ to be a small value, which we choose to be 0.2. In conclusion, our initial parameters are given by
\begin{eqnarray*}
  k^0 &=& \sqrt{\frac{2}{\pi}}\frac{\overline{\delta_t}}{\overline{|\lambda_t|}},\\
  \mu^0 &=& 0.4\left(\sqrt{\frac{\pi}{2}}\overline{\left|\lambda_t\right|}\right),\\
  \alpha_1^0 &=& 0.2\sqrt{\frac{\pi}{2}},\\
  \beta_1^0 &=& 0.2\left(\frac{1}{k^0}\right),\\
  \gamma_1^0 &=& 0.2.
\end{eqnarray*}

\subsection{Simulation}\label{simulation}
We generate four sets of parameters using the initial values $h_{0}=0$ and $r_{0}=E\left(r_{t}\right)$, each of which will result in a weakly stationary Int-GARCH process. The exact parameter values are listed in Table \ref{tab:simu}. Plots of the simulated data are shown in Fig \ref{fig:data_simu}. 

%=======================================================Figure 2===========================================================%
\begin{figure}[ht]
\centering
\includegraphics[ height=1.800in, width=2.300in]{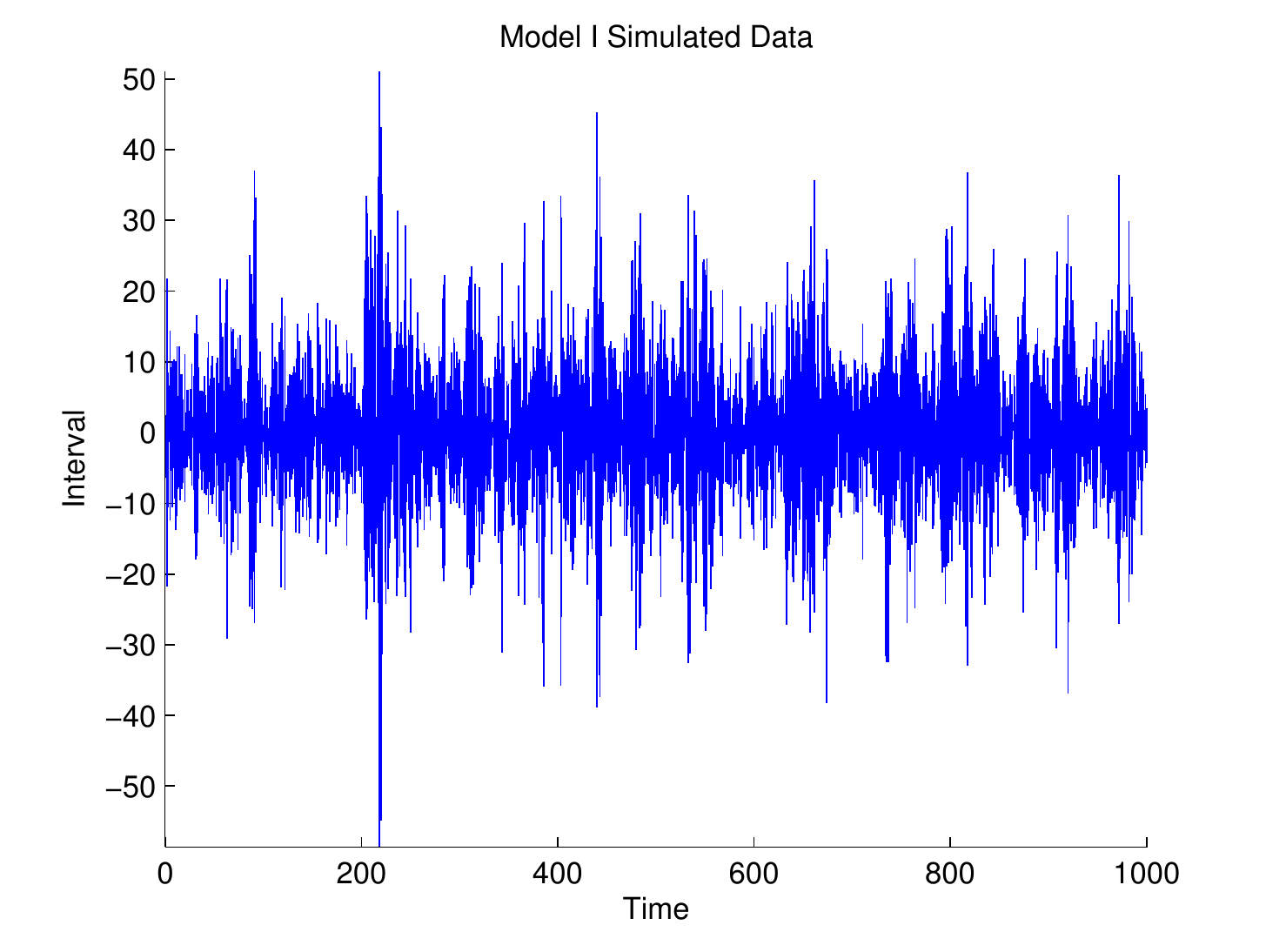}
\includegraphics[ height=1.800in, width=2.300in]{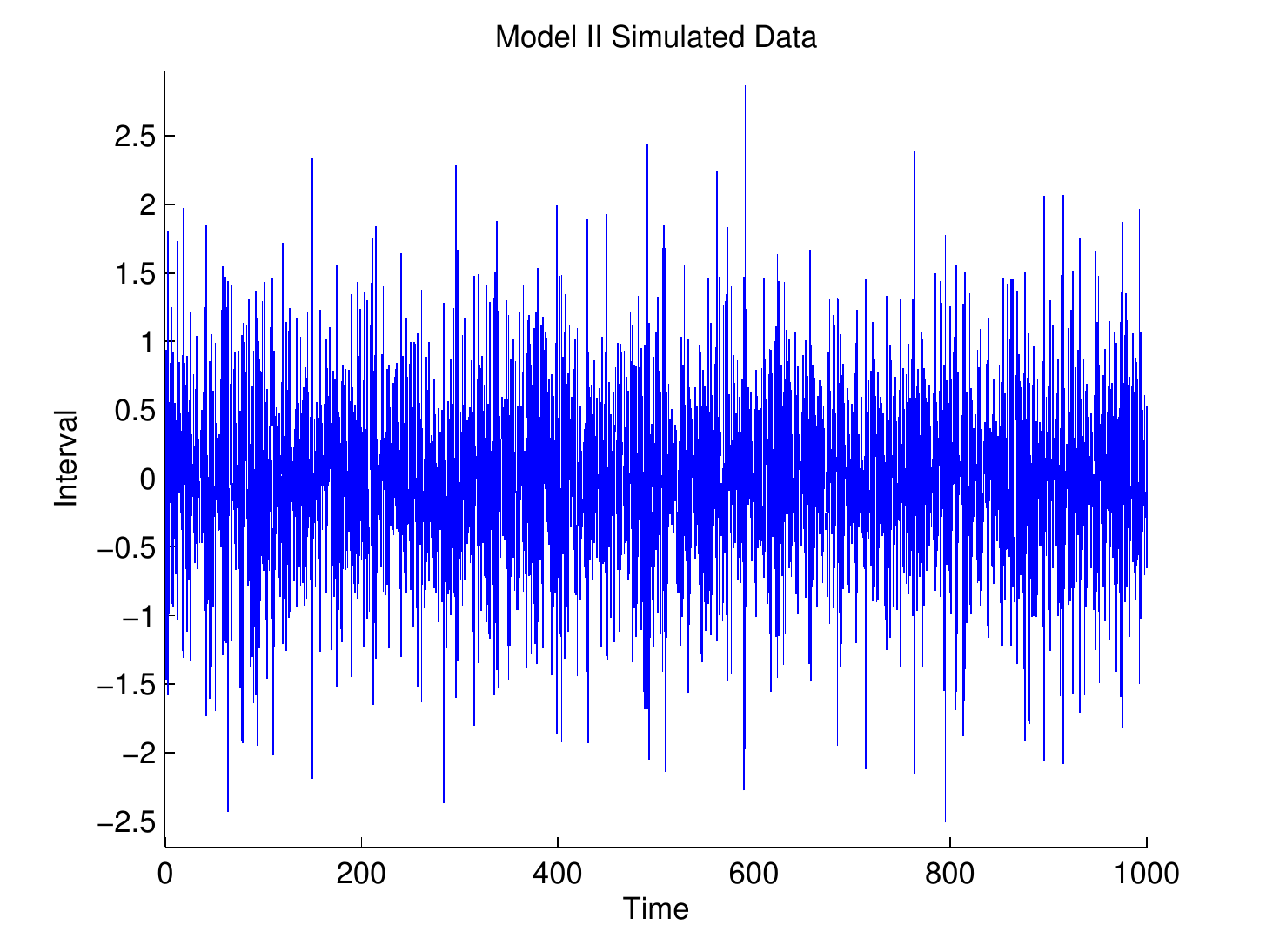}
\caption{Plots of simulated data sets each with $T=1000$.}
\label{fig:data_simu}
\end{figure}

Denote
\begin{eqnarray*}
  \gamma(s) &=& \mbox{Cov}\left(r_{t},r_{t+s}\right),\\
  \gamma_{\lambda}(s) &=& \mbox{Cov}\left(\lambda_{t},\lambda_{t+s}\right),\\
  \gamma_{\delta}(s) &=& \mbox{Cov}\left(\delta_{t},\delta_{t+s}\right).
\end{eqnarray*}
Recall that the theoretical ACF of $\left\{r_t\right\}$ is
\begin{eqnarray*}
\rho(s) & = & \dfrac{\gamma(s)}{\gamma(0)}=\dfrac{\gamma_{\lambda}(s)+\gamma_{\delta}(s)}{\gamma_{\lambda}(0)+\gamma_{\delta}(0)}.
\end{eqnarray*}
We consequently define the sample ACF of $\left\{ r_{t}\right\} $ as
\[
\rho(s)=\dfrac{\hat{\gamma_{\lambda}}(s)+\hat{\gamma_{\delta}}(s)}{\hat{\gamma_{\lambda}}(0)+\hat{\gamma_{\delta}}(0)},
\]
where $\hat{\gamma_{\lambda}}(s)$ and $\hat{\gamma_{\delta}}(s)$ are the sample auto-covariance functions of $\left\{ \lambda_{t}\right\}$ and $\left\{ \delta_{t}\right\} $, respectively. Figure \ref{fig:acf_1} and \ref{fig:acf_2} show the sample ACF's for simulated data sets with 3000 observations from models I and II, respectively.\\

%=======================================================Figure 3===========================================================%
\begin{figure}[ht]
\centering
\includegraphics[ height=1.800in, width=2.300in]{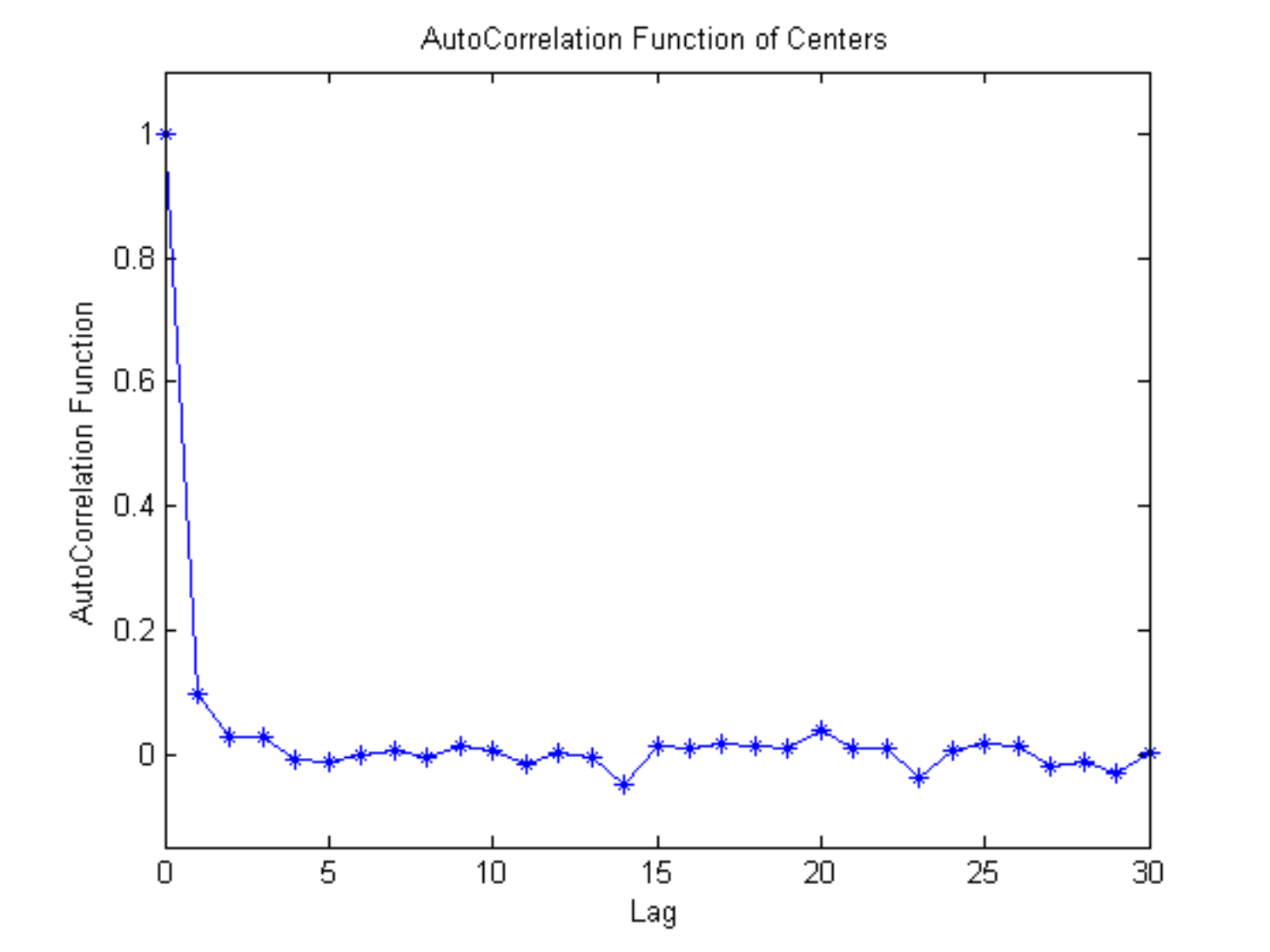}
\includegraphics[ height=1.800in, width=2.300in]{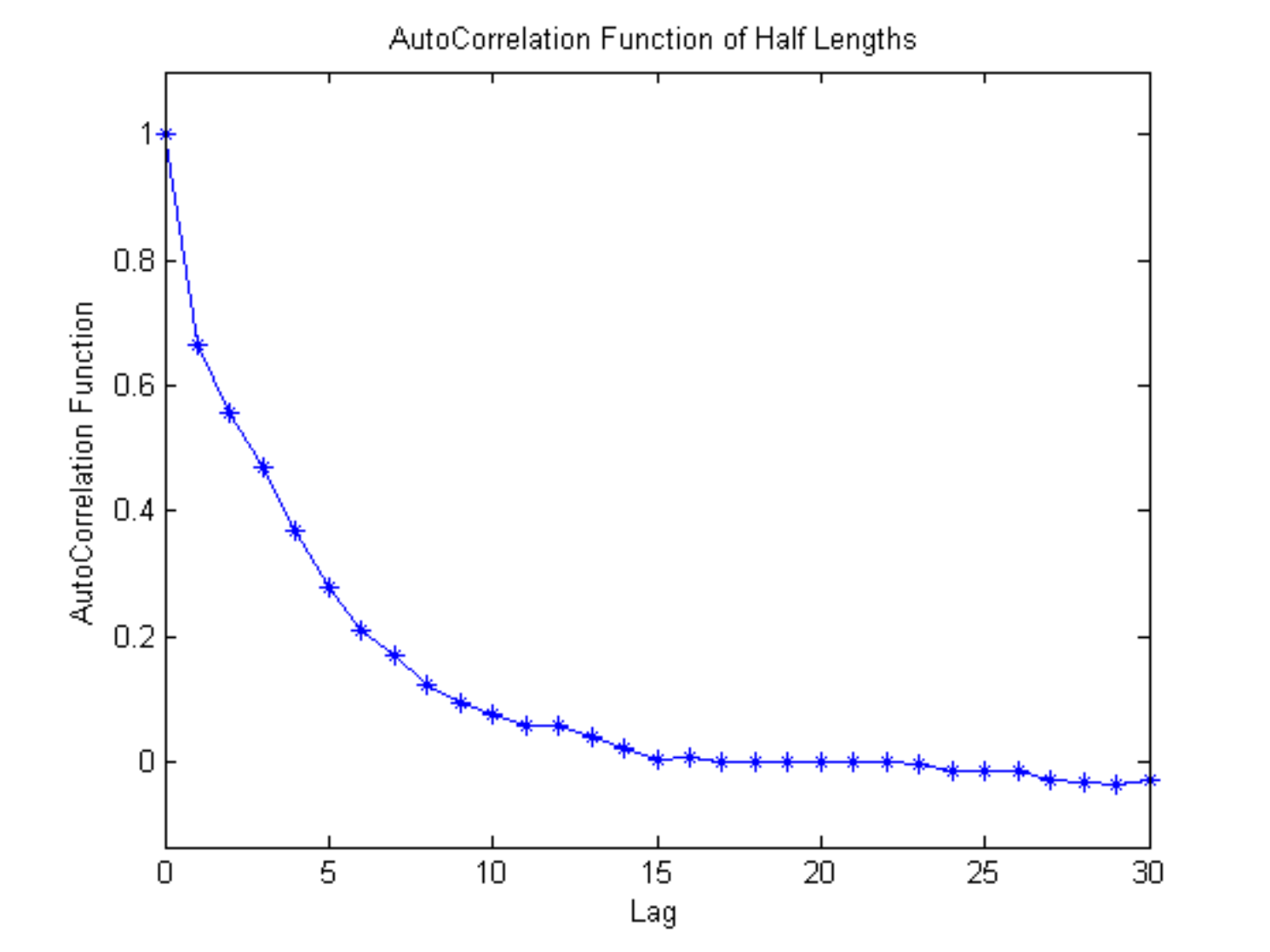}\\
\includegraphics[ height=1.800in, width=2.300in]{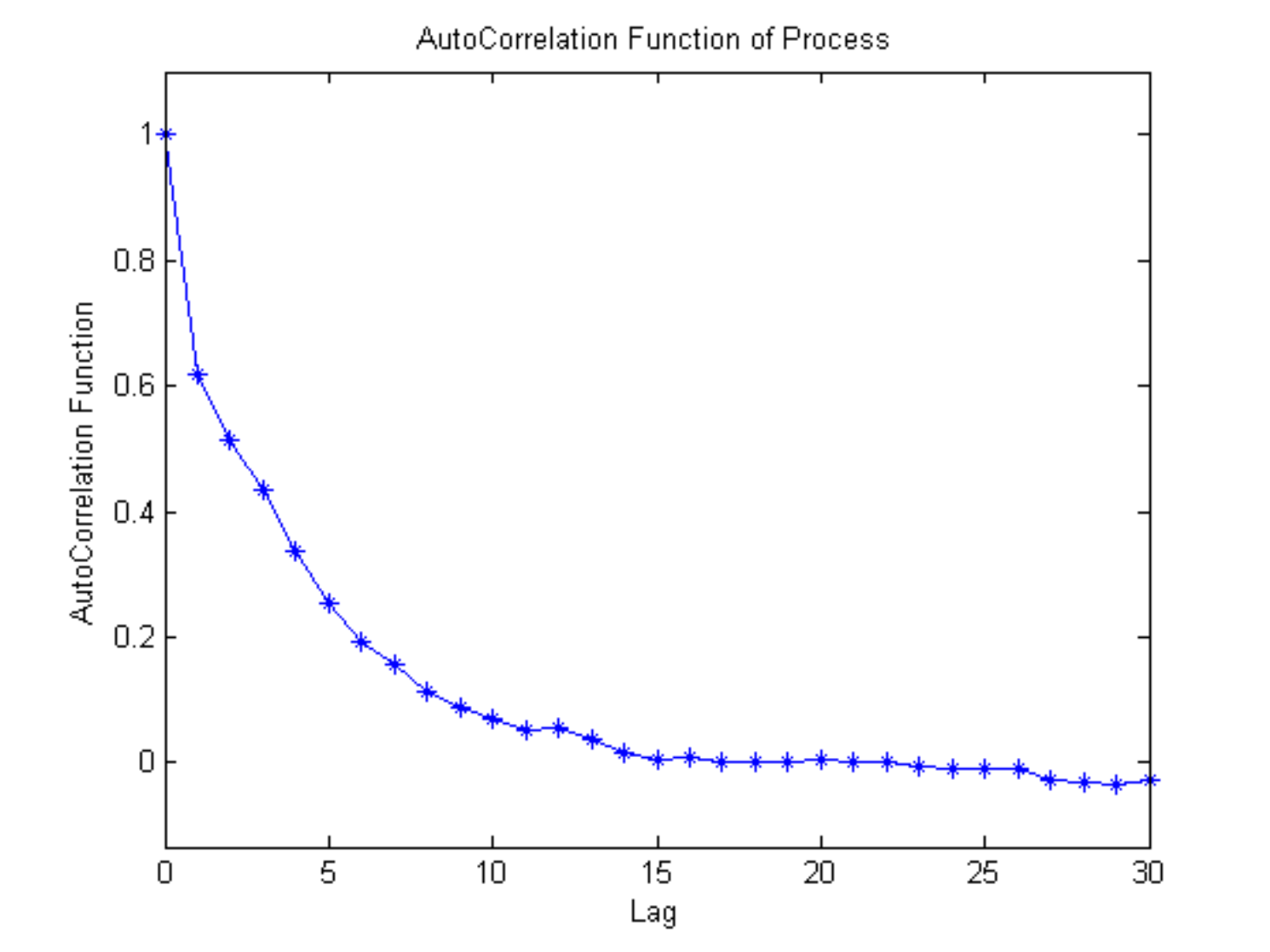}
\caption{Sample auto-correlation functions of a simulated data set from Model I.}
\label{fig:acf_1}
\end{figure}

%=======================================================Figure 4===========================================================%
\begin{figure}[ht]
\centering
\includegraphics[ height=1.800in, width=2.300in]{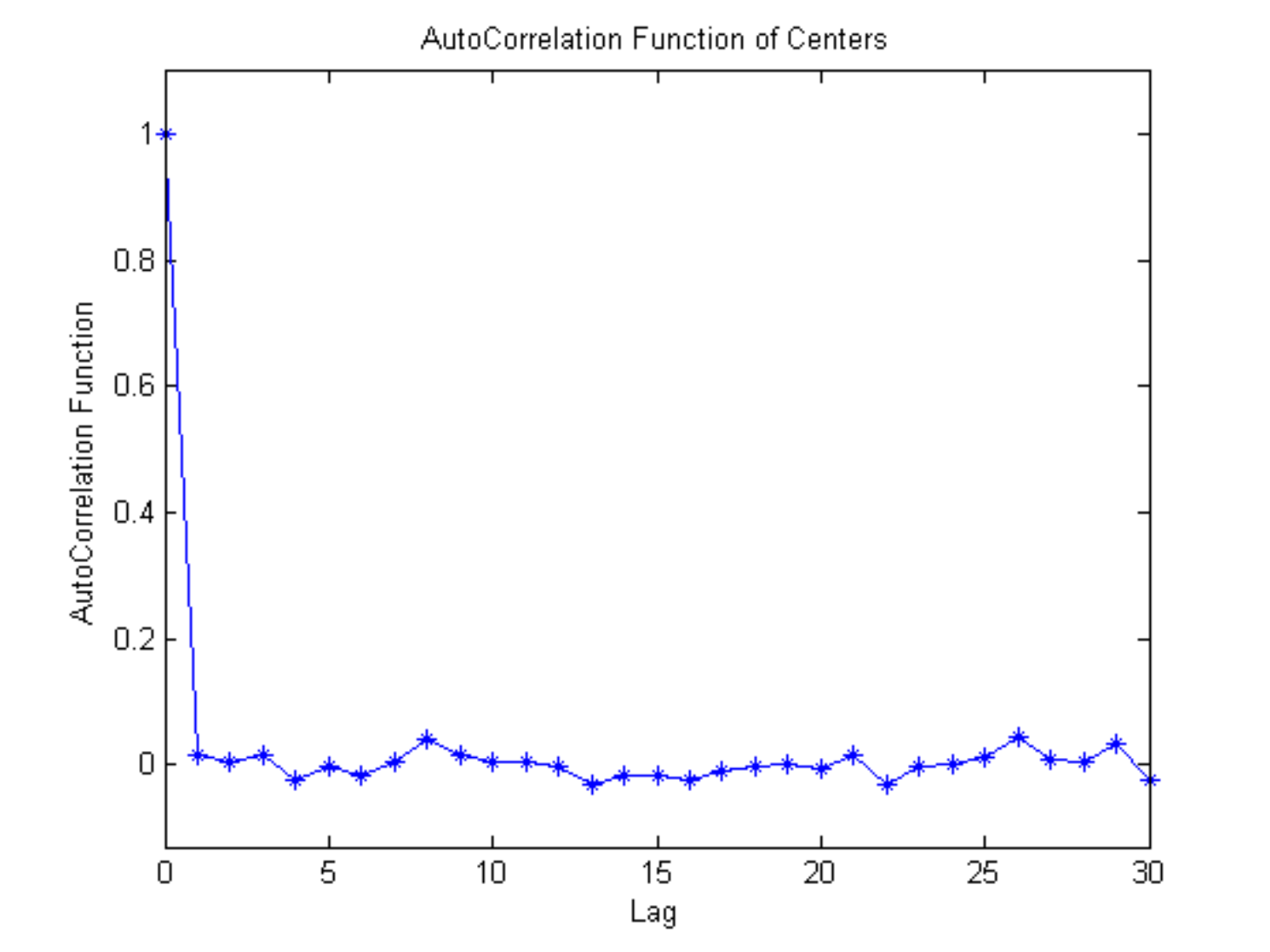}
\includegraphics[ height=1.800in, width=2.300in]{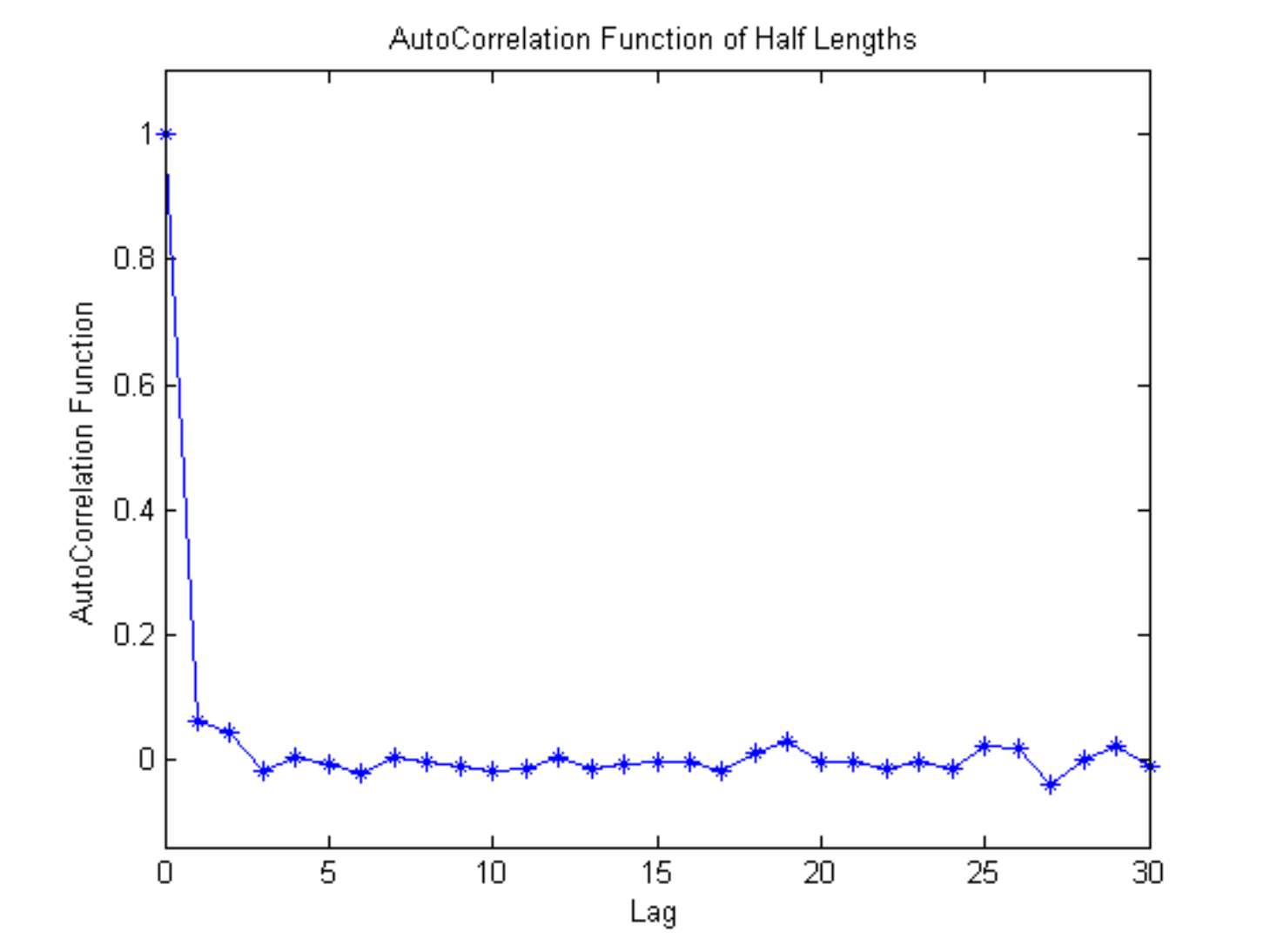}\\
\includegraphics[ height=1.800in, width=2.300in]{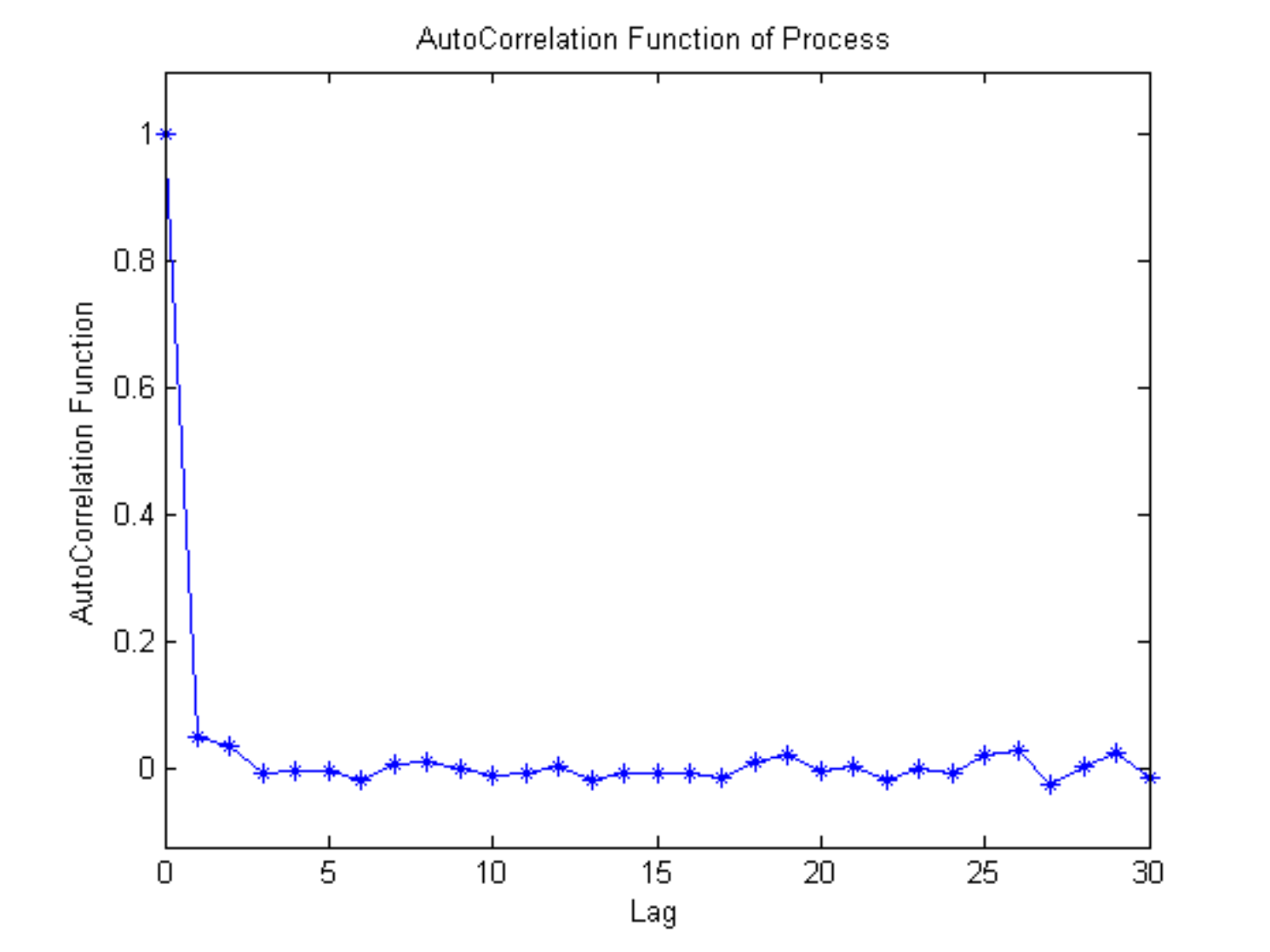}
\caption{Sample auto-correlation functions of a simulated data set from Model II.}
\label{fig:acf_2}
\end{figure}

\indent For each of the four models, we simulate a data set with 3000 observations and calculate the CLSE of the parameters using the Newton-Raphson algorithm. We found from our experience that excluding $k$ from the algorithm to search for the CLSE for $\left\{\mu, \alpha_1, \beta_1, \gamma_1\right\}$ results in the algorithm being more stable. Since the initial estimate for $k$ as described in Section \ref{sub:Initialization-of-the} turns out to be good enough, it is used as the final estimate. An alternative strategy is to update k and the rest parameters in separate steps. We repeat the process of simulation and estimation for 100 times independently, and the overall results are reported in Table \ref{tab:simu}. We see that our conditional least squares estimates are very close to the true parameters with small empirical standard errors.

%=======================================================Table 1=======================================================%
\begin{table}
\center
\caption{Average Result of 100 Repetitions}
\label{tab:simu}
\bigskip
\begin{tabular}{cccccc}
\toprule 
Model  & \multicolumn{2}{c}{Parameters} & Mean Estimate  & Mean Bias  & Empirical Standard Error \tabularnewline
\midrule 
 & $k$  & 4.7162  & 4.7173  & 0.0677  & 0.0832\tabularnewline
 & $\mu$  & 0.4724  & 0.4917  & 0.0671  & 0.0842\tabularnewline
I & $\alpha_{1}$  & 0.2637  & 0.2664  & 0.0206  & 0.0251\tabularnewline
 & $\beta_{1}$  & 0.0906  & 0.0887  & 0.0055  & 0.0063\tabularnewline
 & $\gamma_{1}$  & 0.1796  & 0.1778  & 0.0383  & 0.0475\tabularnewline
\midrule 
 & $k$  & 2.7330  & 2.7370  & 0.0396  & 0.0491\tabularnewline
 & $\mu$  & 0.1385  & 0.1397  & 0.0110  & 0.0139\tabularnewline
II & $\alpha_{1}$  & 0.2572  & 0.2621  & 0.0180  & 0.0222\tabularnewline
 & $\beta_{1}$  & 0.0202  & 0.0192  & 0.0059  & 0.0073\tabularnewline
 & $\gamma_{1}$  & 0.1459  & 0.1398  & 0.0521  & 0.0651\tabularnewline
\midrule 
 & $k$  & 5.4871  & 5.5012  & 0.0672  & 0.0908\tabularnewline
 & $\mu$  & 0.5331  & 0.5359  & 0.0453  & 0.0574\tabularnewline
III & $\alpha_{1}$  & 0.1782  & 0.1751  & 0.0127  & 0.0154\tabularnewline
 & $\beta_{1}$  & 0.0253  & 0.0254  & 0.0027  & 0.0036\tabularnewline
 & $\gamma_{1}$  & 0.1396  & 0.1364  & 0.0538  & 0.0669\tabularnewline
\midrule 
  & $k$  & 1.9108  & 1.9103  & 0.0286  & 0.0358\tabularnewline
 & $\mu$  & 0.3640  & 0.3654  & 0.0384  & 0.0458\tabularnewline
IV & $\alpha_{1}$  & 0.2642  & 0.2652  & 0.0211  & 0.0269\tabularnewline
 & $\beta_{1}$  & 0.0228  & 0.0216  & 0.0083  & 0.0101\tabularnewline
 & $\gamma_{1}$  & 0.0705  & 0.0704  & 0.0745  & 0.0884\tabularnewline
\bottomrule
\end{tabular}
\end{table}

%=======================================================================================================================%

%===================================================Real data analysis======================================================%
\section{Real data analysis}
%========================================================================================================================%
In this section, we apply our Int-GARCH model to analyze the Dow Jones Industrial Average Index component stocks. We obtained the daily data of totally 3019 trading days, from January 3, 2000 to December 31, 2011. Figure \ref{fig:plot-return} shows plots of the return range data for 4 randomly selected stocks: BA, JPM, KO, and TRV. To reveal more information in these data, we highlight the days when the return range is long (above $75\%$ quantile of the entire data) but the average return is small (absolute value of the center is below $25\%$ quantile). We see that most of such days occurred in the years 2000-2004 and 2008-2010. As we discussed before, these days show large variability in price, but the point-valued volatility models such as GARCH tend to underestimate it. The sample ACF of the BA return range is displayed in Figure \ref{fig:acf-BA}. It is very similar to the theoretical ACF of model I in Figure \ref{fig:realACF_1} from Section \ref{simulation}, which indicates the feasibility of our Int-GARCH model. We estimate the parameters by the CLS method we proposed in Section \ref{CLSE} and list the results for 10 randomly selected stocks in Table \ref{tab:est_model}. The fitted models show some patterns. 1) $\alpha_1$ is much smaller than $\beta_1$ in magnitude, indicating that the return range is of much more importance than a ``snapshot" return, in terms of their contributions to the volatility. 2) $\gamma_1$ is either very small or exactly 0; most likely an Int-ARCH model is sufficient here.

%===================================================Figure 5: return range plot=====================================================%
\begin{figure}[ht]
\centering
\includegraphics[ height=1.800in, width=2.300in]{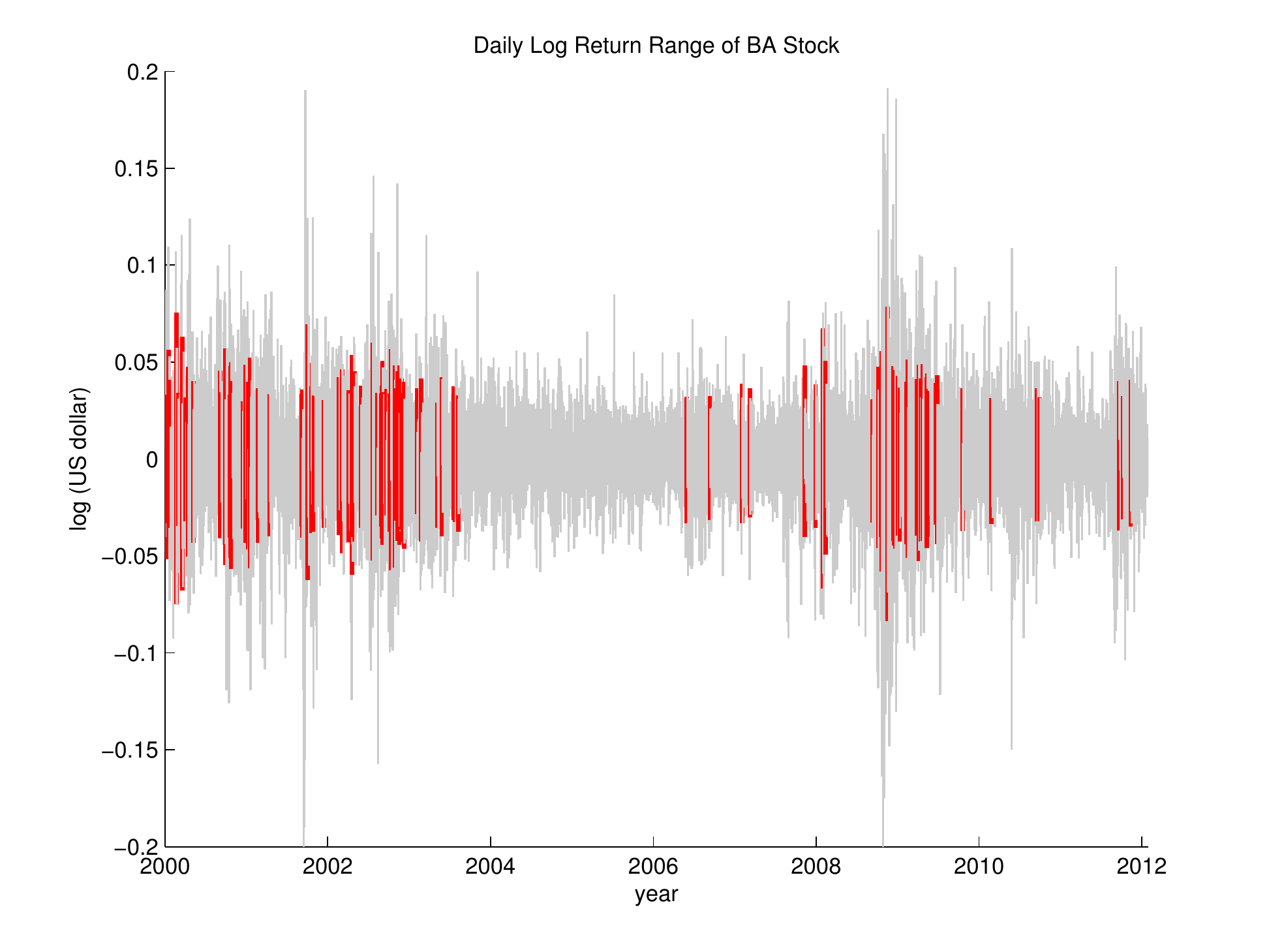}
\includegraphics[ height=1.800in, width=2.300in]{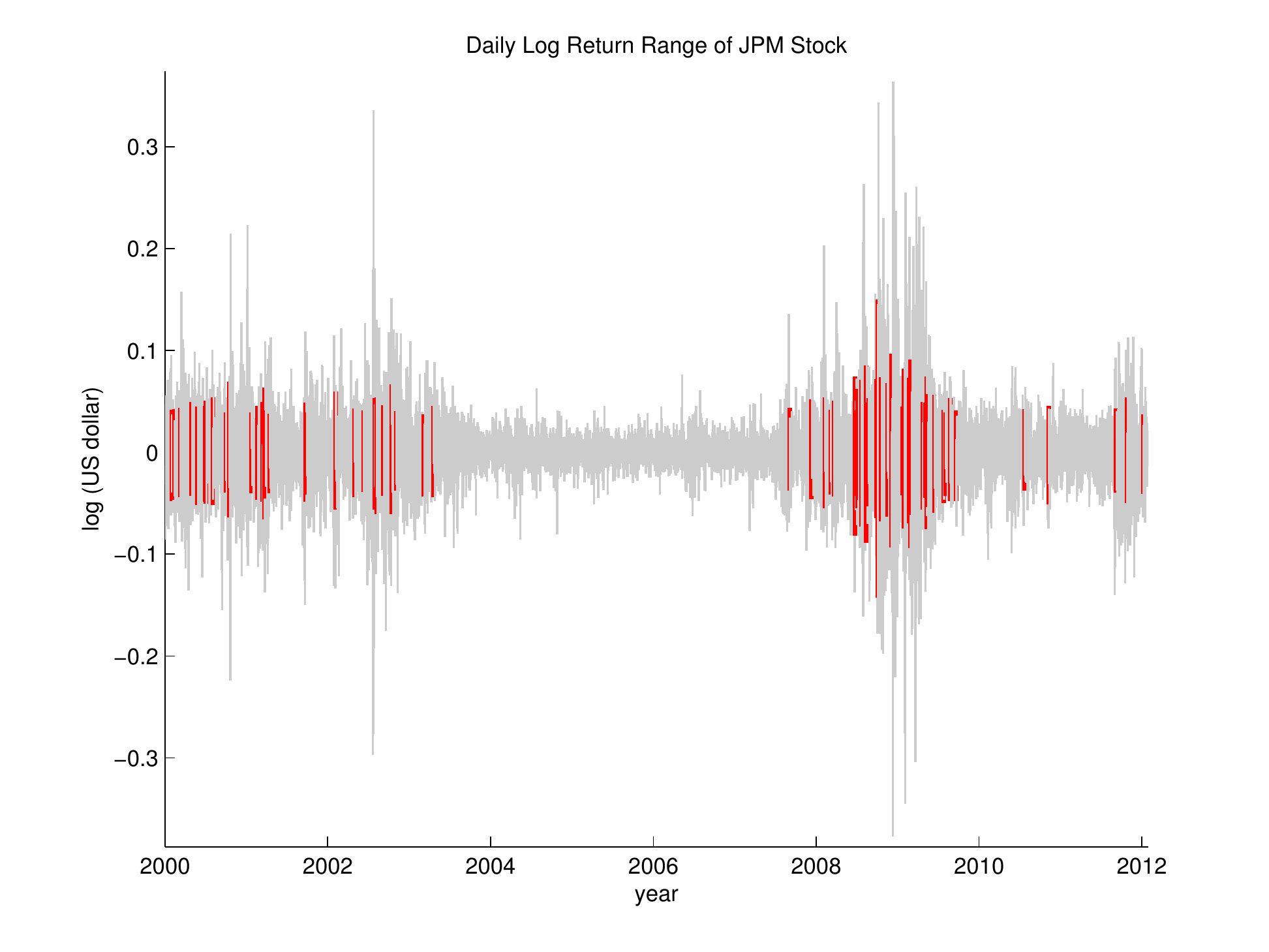}\\
\includegraphics[ height=1.800in, width=2.300in]{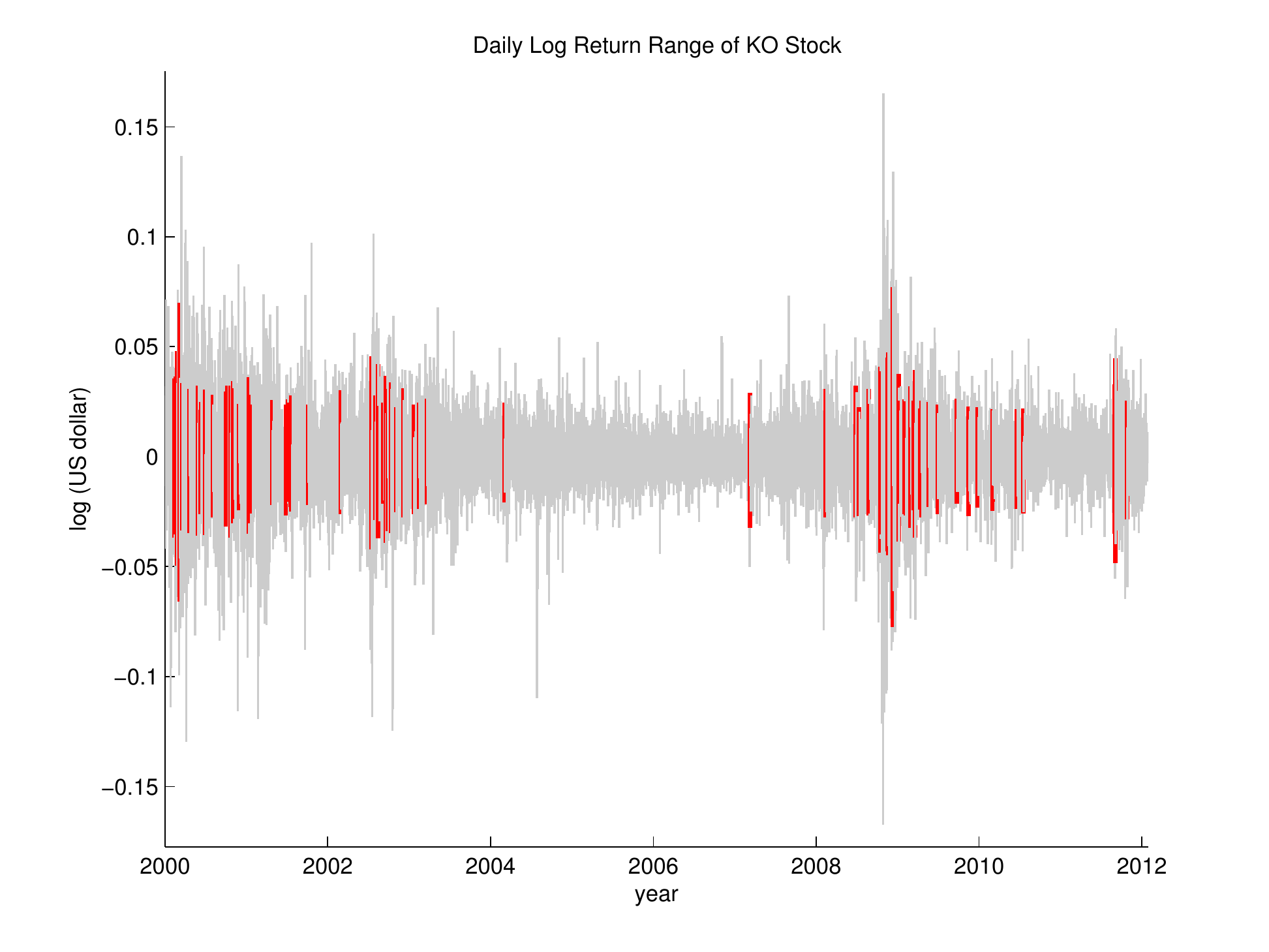}
\includegraphics[ height=1.800in, width=2.300in]{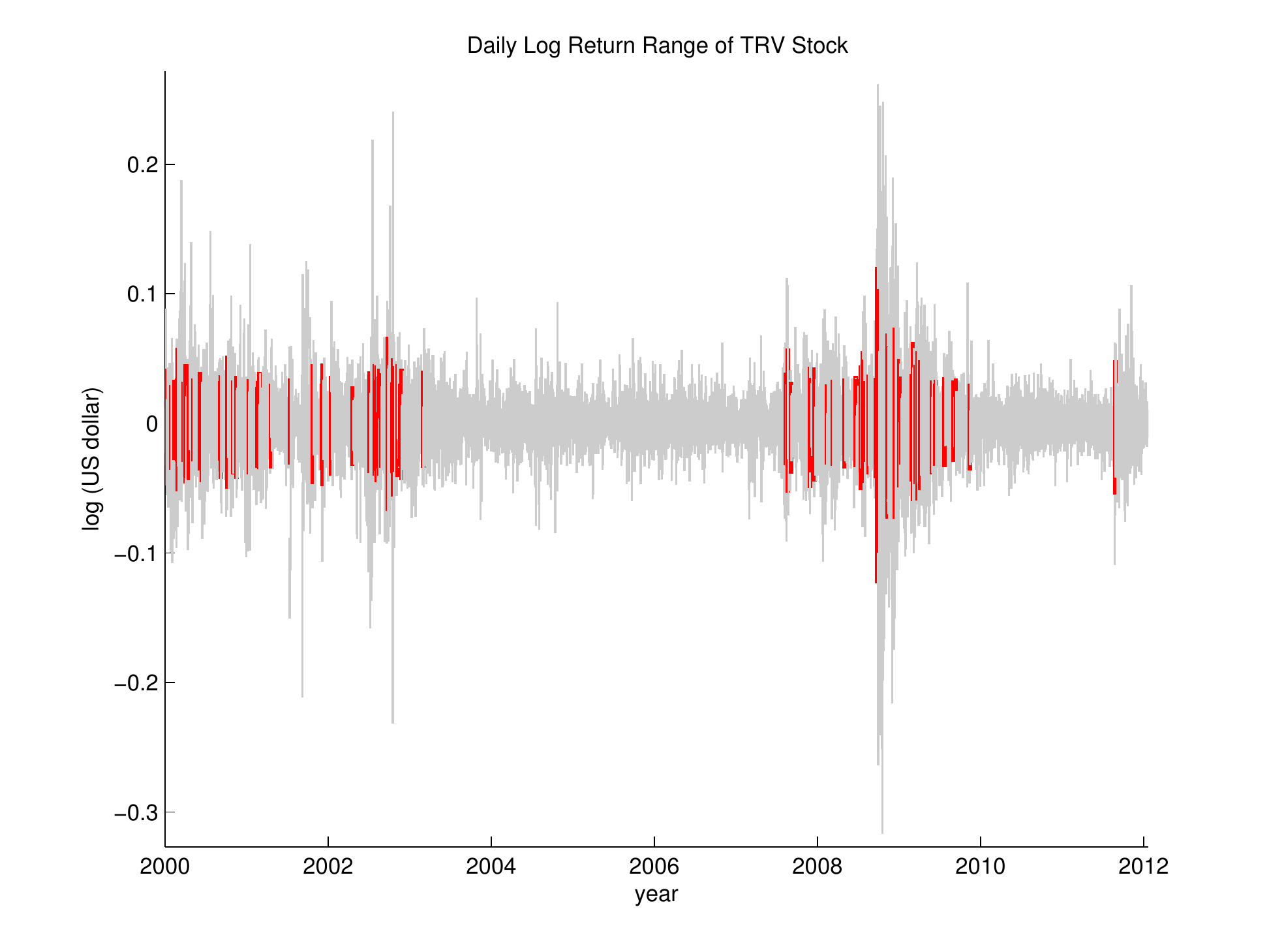}\\
\caption{Plots of the daily return range data for selected stocks from Dow Jones Index.}
\label{fig:plot-return}
\end{figure}

%==================================================Figure 6: sample ACF of BA stock=================================================%
\begin{figure}[ht]
\centering
\includegraphics[ height=1.800in, width=2.300in]{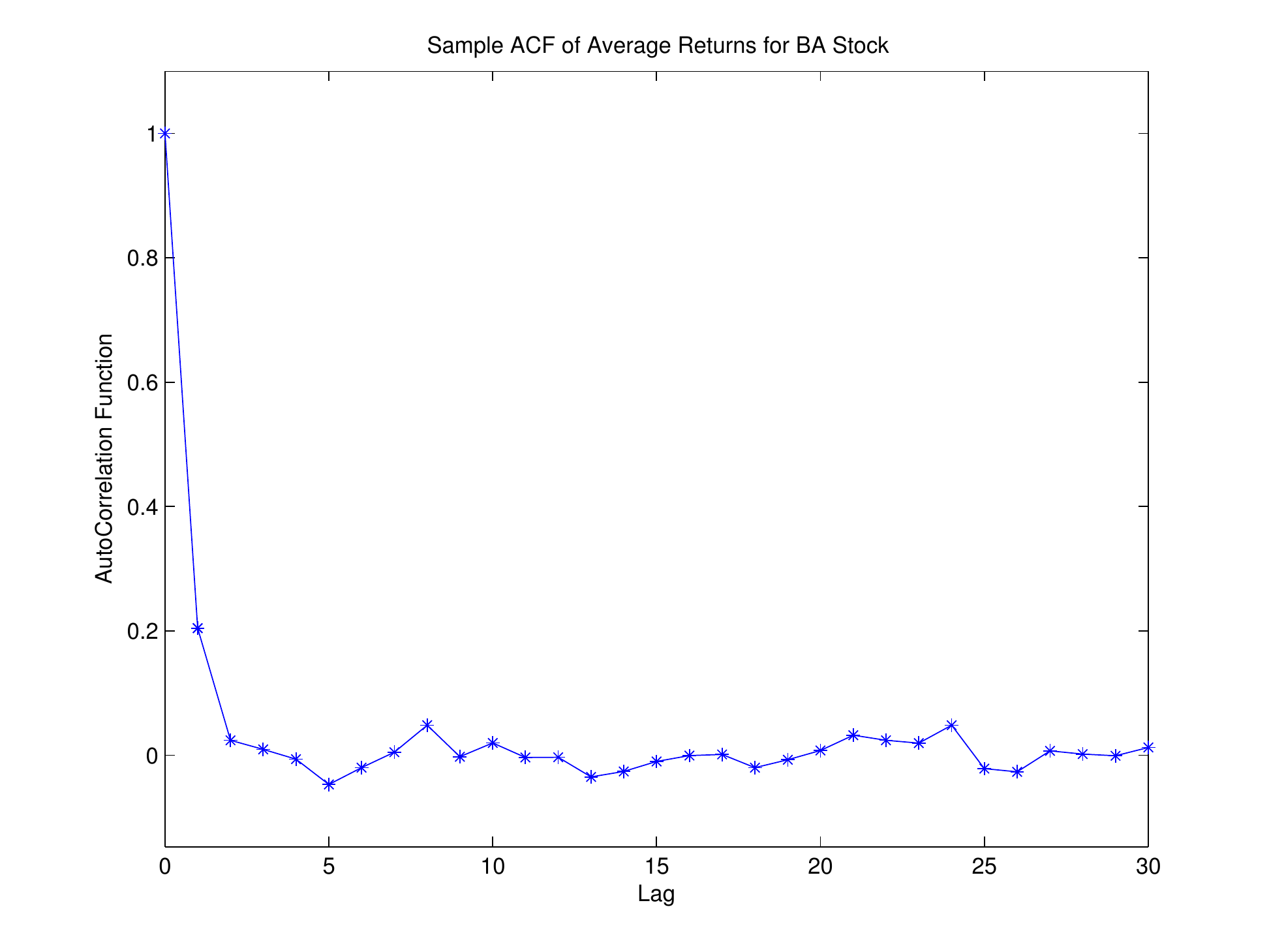}
\includegraphics[ height=1.800in, width=2.300in]{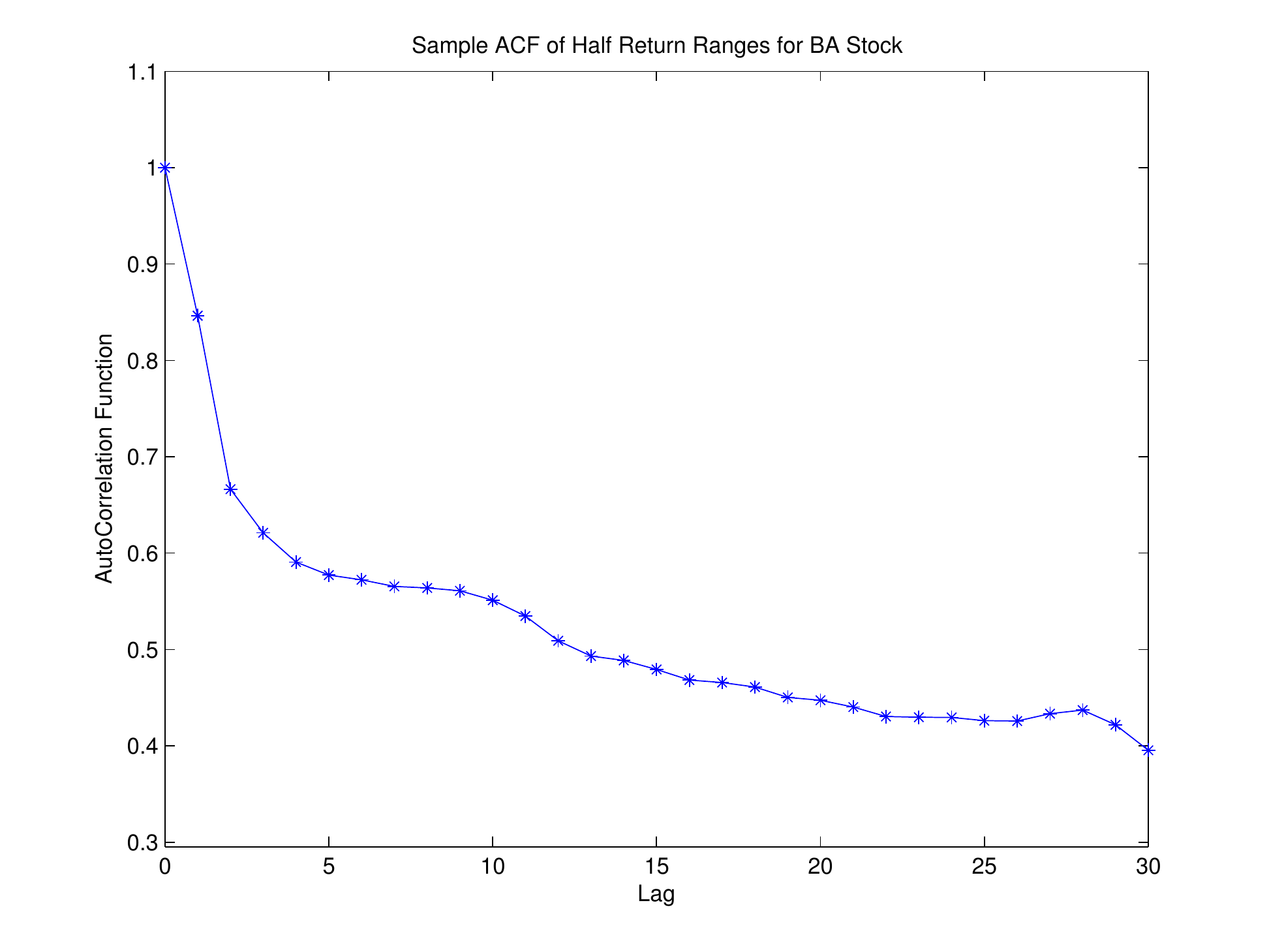}\\
\includegraphics[ height=1.800in, width=2.300in]{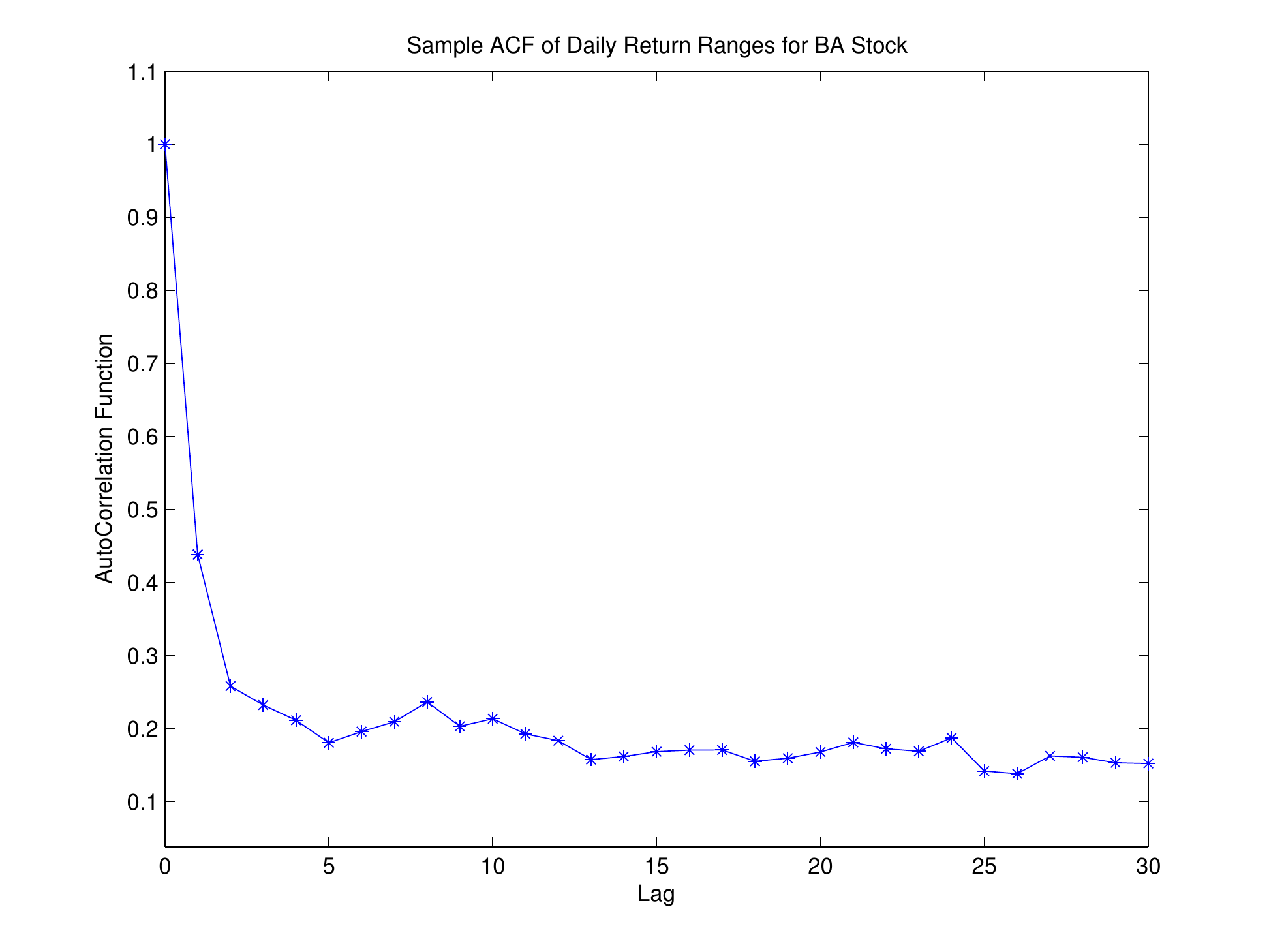}
\caption{Sample auto-correlation function of BA stock.}
\label{fig:acf-BA}
\end{figure}

%=================================================Table 2: estimated Int-GARCH models===========================================%
%===========================================================================================================================%
\begin{table}[ht]
\center
\caption{Estimated Int-GARCH parameters for 10 selected stocks.}
\label{tab:est_model}
\bigskip
\begin{tabular}{rrrrrrrr}
%\hline
%                      \multicolumn{ 8}{c}{Estimated I-GARCH Models for Daily Stock Price Range Data} \\
%
%          &            &            &            &            &            &            &            \\
\hline
           &            &            &                       \multicolumn{ 5}{c}{Parameter Estimates} \\

     Index &      Stock &            &          k &       $\mu$ &   $\alpha_1$ &    $\beta_1$ &   $\gamma_1$ \\
\hline
           &            &            &            &            &            &            &            \\

 Dow Jones &        AXP &            &     1.7524 &     0.0016 &     0.0635 &     0.4884 &     0.0072 \\

           &         BA &            &     1.6828 &     0.0025 &     0.0686 &     0.4687 &          0 \\

           &        BAC &            &     1.5909 &     0.0017 &     0.0762 &     0.5337 &     0.0084 \\

           &         DD &            &     1.7744 &     0.0019 &     0.0209 &     0.4817 &          0 \\

           &        JPM &            &     1.6873 &     0.0018 &     0.0379 &     0.5195 &     0.0021 \\

           &         KO &            &      1.8150 &     0.0013 &     0.0383 &      0.4620 &          0 \\

           &       MSFT &            &     1.7513 &     0.0017 &      0.0320 &     0.4878 &          0 \\

           &         T  &            &     1.8106 &     0.0021 &     0.0524 &     0.4475 &     0.0001 \\

           &        TRV &            &     1.8167 &     0.0015 &     0.0092 &     0.4877 &          0 \\

           &        WMT &            &     1.8415 &     0.0015 &     0.0376 &      0.4590 &          0 \\
\hline
\end{tabular}
\end{table}

As we mentioned above, most of the trading days with large return ranges but small average returns (i.e., around the years 2002 and 2008) happened during the periods of market crash. This is consistent with one's intuition: investors tend to behave irrationally and take swift actions frequently in the time of crisis, resulting in unusually large variability of the price. We are particularly interested in these periods when GARCH and our Int-GARCH models tend to differ significantly. In figure \ref{fig:int-garch-compare}, we plot the estimated volatility based on both models for two randomly selected stocks for each of the two periods. They are quite consistent in terms of the overall trend, but our Int-GARCH curve shows much more detailed fluctuations than the smoother GARCH curve. And, just as we anticipated, the GARCH-based volatility is generally smaller than our Int-GARCH-based volatility.  

%================================================Figure 7: Int-GARCH v.s. GARCH=================================================%
\begin{figure}[ht]
\centering
\includegraphics[ height=2.000in, width=5.400in]{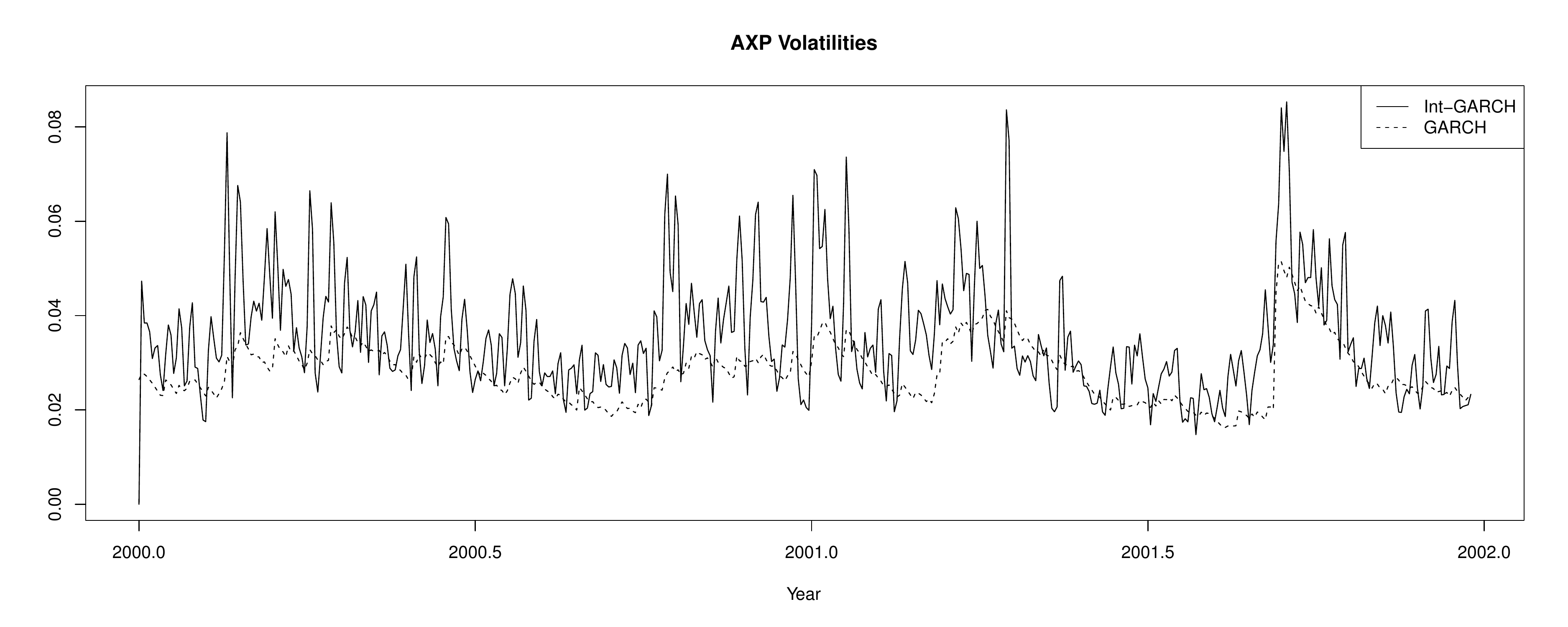}\\
\includegraphics[ height=2.000in, width=5.400in]{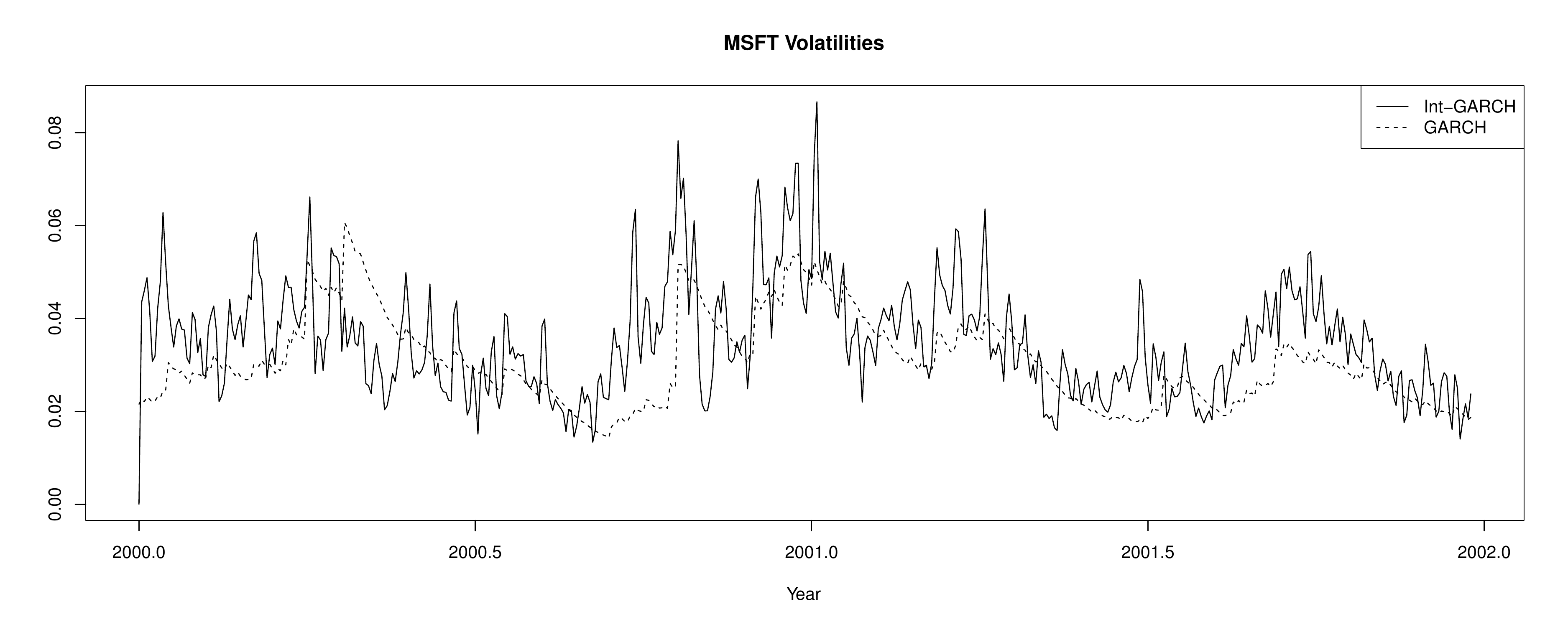}\\
\includegraphics[ height=2.000in, width=5.400in]{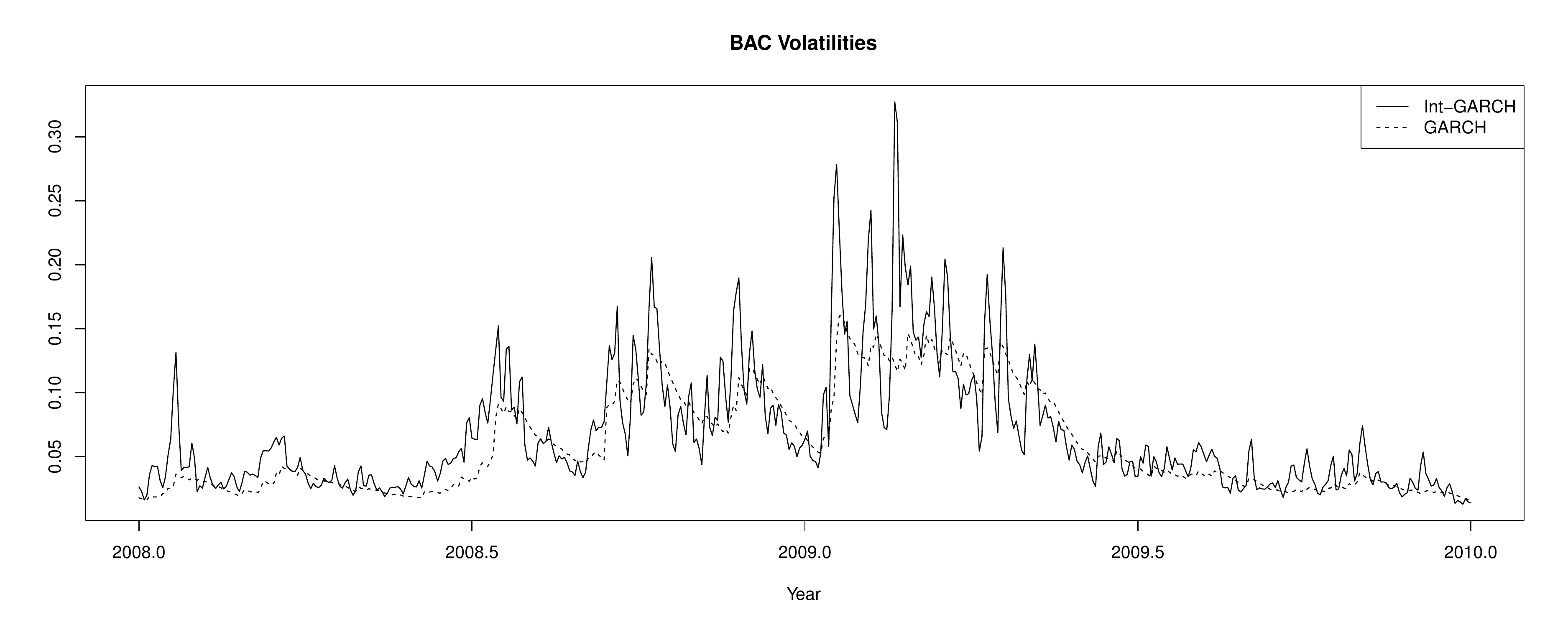}\\
\includegraphics[ height=2.000in, width=5.400in]{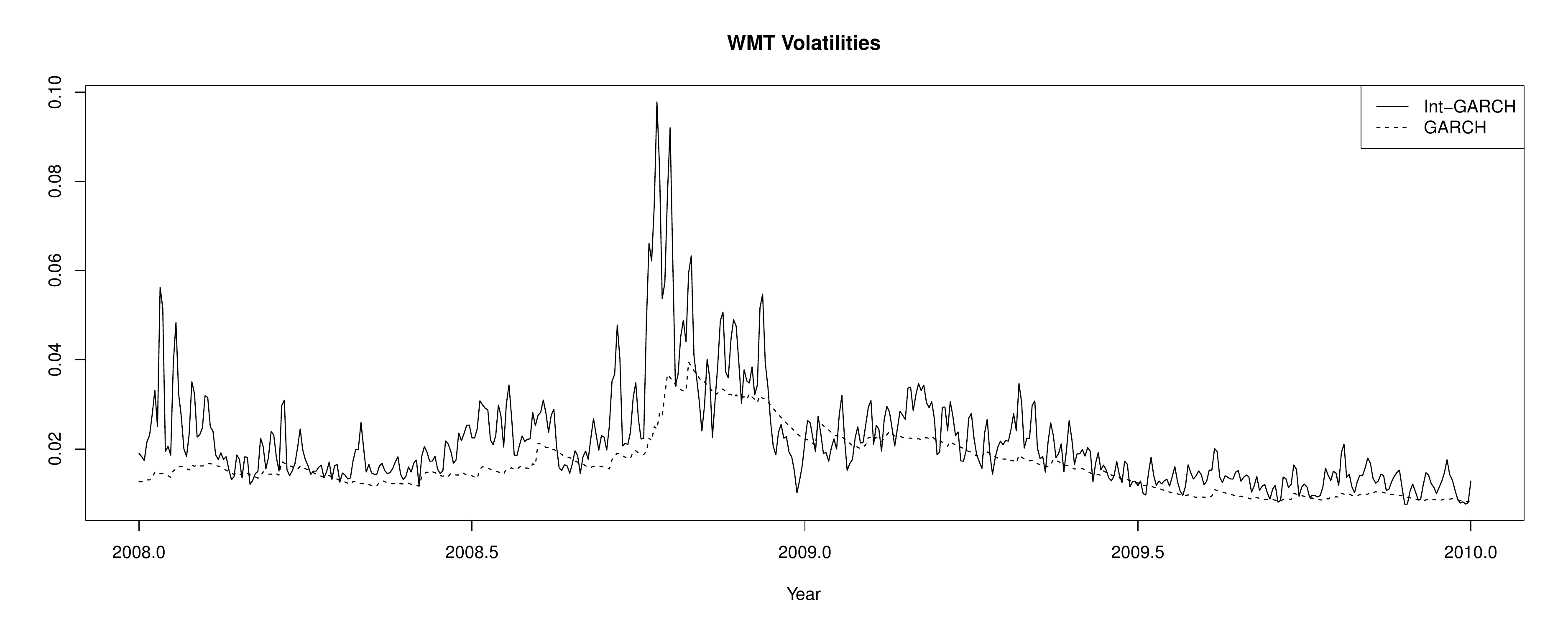}
\caption{Comparison of volatility estimation from Int-GARCH and GARCH models for selected stocks for the recession periods 2000-2002 and 2008-2010.}
\label{fig:int-garch-compare}
\end{figure}

It is hard to say which estimate is better in the plots of Figure \ref{fig:int-garch-compare}. In principle, the quality of volatility modeling is only judged by its usage in the financial applications. So our main purpose here is to analyze and interpret the unique characteristics of our Int-GARCH process, so as to provide practical guidance. To this end, we further compare our Int-GARCH estimate to the realized volatility (RV) with 5-min sampling frequency. Both of these two methods use intraday data to estimate volatility, but they have subtle differences. RV is a model free estimator, and possesses very nice theoretical properties such as strong consistency (e.g., Andersen et al. 2001, Barndorff-Nielsen and Shephard 2002). It uses all or a large portion of the intraday transaction data, and therefore is expected to have a ``fair" reflection of the intraday price variability. However, because of the microstructure noises, it usually cannot achieve its optimal performance. For example, it is generally believed that RV with moderate sampling frequency (e.g., 5-min) has a downward bias (e.g., Bandi and Russell 2006, 2008). Our Int-GARCH model reads in the intraday information through the high-low range. As we discussed before, although range is only a summary statistics, it does contain rich information about the intraday activities. So Int-GARCH model is expected to inherit to a considerable degree the advantage of RV -revealing detailed intraday variability. What is important, it is robust to the microstructure noises, thereby avoiding the intrinsic (downward) bias associated with RV. 

We display the comparative plots of Int-GARCH and RV for two stocks, AXP and WMT, for the period 1/2/2000-12/31/2001, in Figure \ref{fig:int-garch-rv}. The two estimates show very similar fluctuation patterns, the RV being smaller than the Int-GARCH in general most likely due to its downward bias. Particularly, we notice a few interesting trading dates from these plots to illustrate the advantages of Int-GARCH. On April 26, 2001, the RV of AXP shows a very large spike of \$0.127, while the Int-GARCH and GARCH models give much smaller estimates of \$0.031 and \$0.035, respectively. The transaction data on that day reveals that the price mostly fluctuated around \$41-42 for the whole day except for about half an hour in the afternoon it suddenly dropped to \$38.43 but then quickly rose back to \$41.90. The same situation happened to the WMT stock too on April 9, 2001, on which day the price suddenly soared from the stable range of \$50-\$51 to a much higher level of \$82.5487 but only stayed there for 5 minutes. Such a lightening change of price usually is not due to the market volatility, so for this case RV is probably being too sensitive and giving a ``false alarm''. Separately, another interesting phenomenon is observed for AXP stock on March 10, 2000 and for WMT stock on April 4, 2000, respectively. For these two cases, RV and Int-GARCH estimate are relatively close (\$0.059 and \$0.048 for AXP, and \$0.063 and \$0.076 for WMT, respectively), but both are much higher than the GARCH estimate (\$0.030 and \$0.039 for AXP and WMT, respectively). To get insight into this, we plot the corresponding 5-min intraday prices for the entire days in Figure \ref{fig:intraday-data}. In both plots, some significant fluctuations are shown in the later portion of the day, compared to which the closing-to-closing returns are quite small. So clearly the return-based GARCH model is underestimating the volatility here, and the RV and Int-GARCH estimate are more fair. There is yet another situation when Int-GARCH shows typical advantages. On January 22, 2001, the GARCH estimated volatility for the AXP stock is \$0.037, while RV and Int-GARCH yield \$0.072 and \$0.074, respectively. We list the adjusted closing prices for that day as well as a few days nearby in Table \ref{tab:axp-close}. Notice that the return on January 22 is not really small, but the nearby returns are. So we reasonably conjecture that GARCH in this case is trapped in the ``volatility clustering" mechanism and too rigid to capture the high volatility that is not in a ``cluster".

%====================================================Figure 8: Int-GARCH v.s. RV=======================================================%
\begin{figure}[ht]
\centering
\includegraphics[ height=2.400in, width=5.800in]{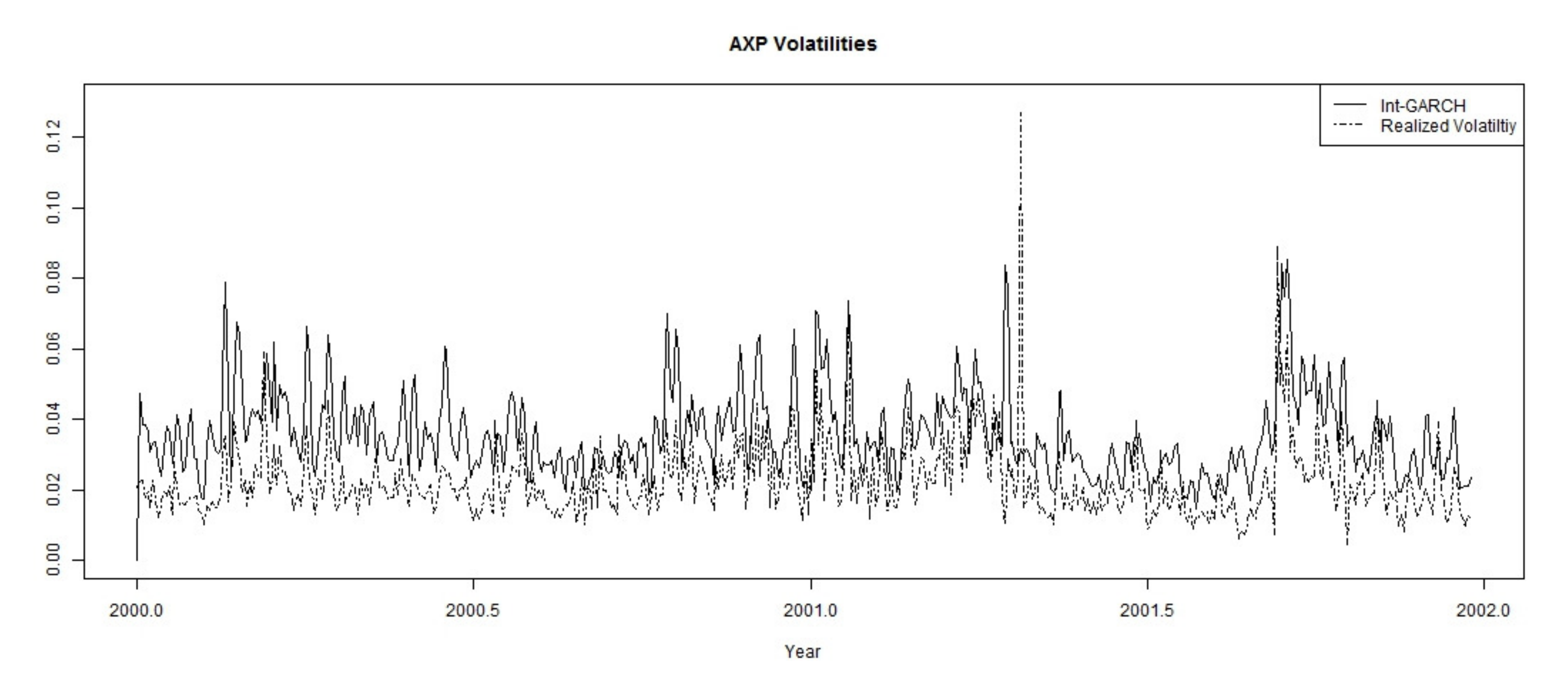}
\includegraphics[ height=2.400in, width=5.800in]{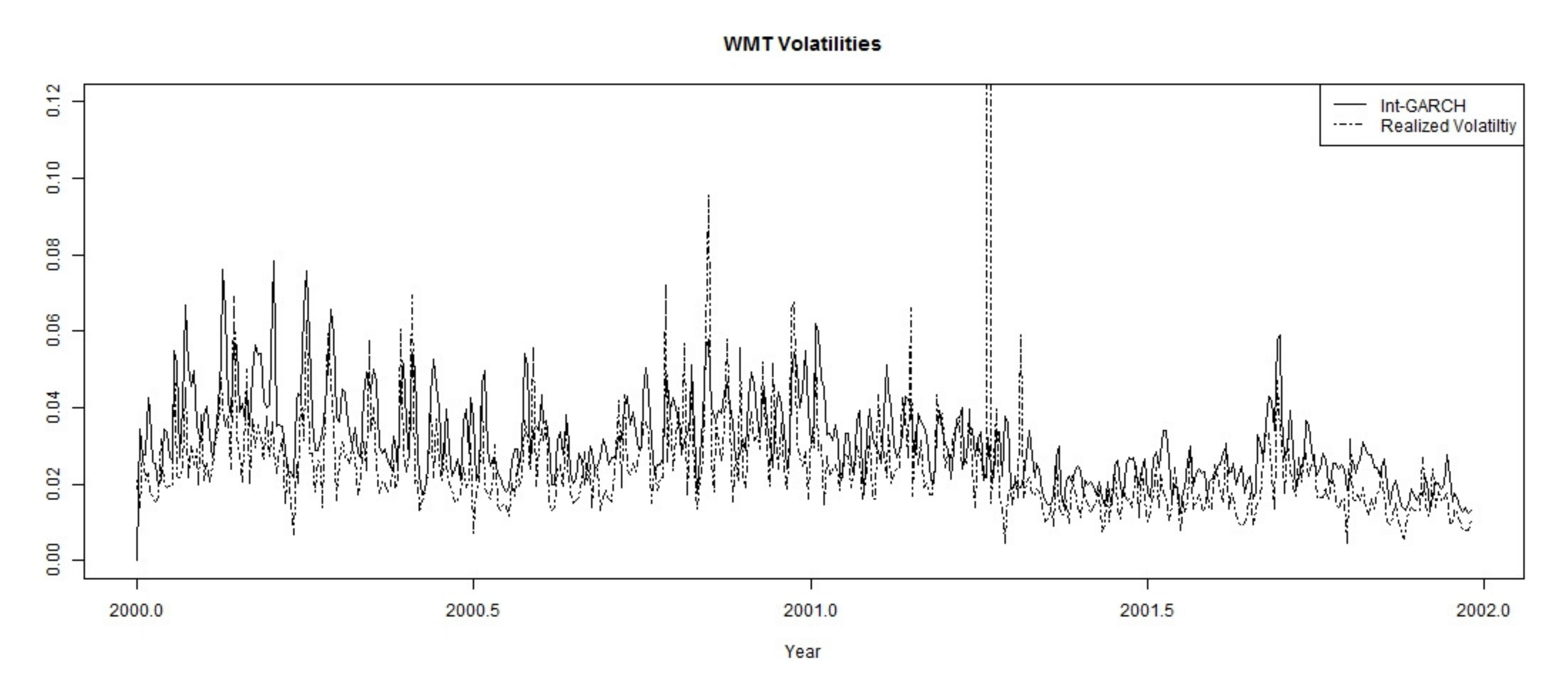}\\
\caption{Plots of the Int-GARCH estimated volatility versus RV  for AXP and WMT stocks for the period 1/2/2000 - 12/31/2001.}
\label{fig:int-garch-rv}
\end{figure}

%====================================================Figure 9=======================================================================%
\begin{figure}[ht]
\centering
\includegraphics[ height=2.000in, width=2.600in]{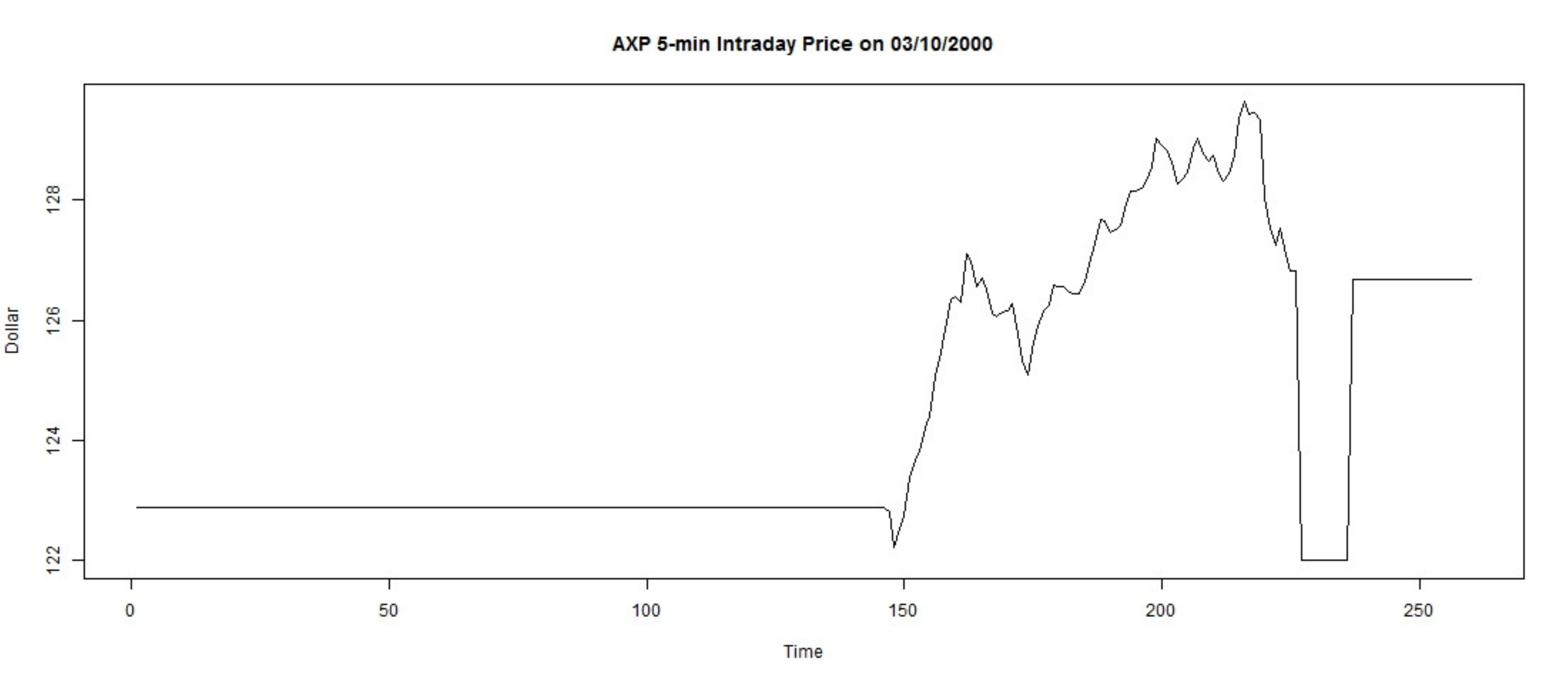}
\includegraphics[ height=2.000in, width=2.600in]{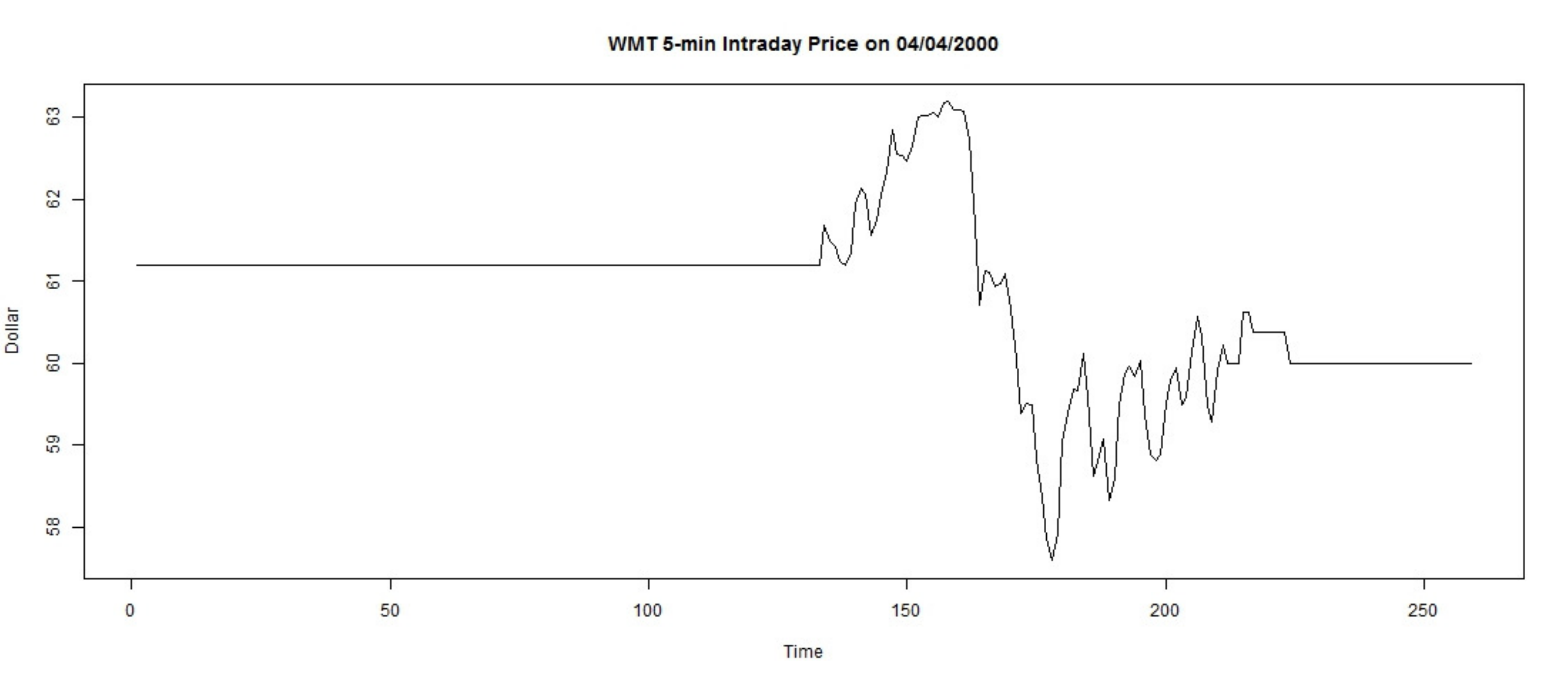}
\caption{Plots of the 5-min intraday price for AXP stock on March 10, 2000 and WMT stock on April 4, 2000, respectively.}
\label{fig:intraday-data}
\end{figure}

%===================================================Table: AXP closing prices====================================================%
%===========================================================================================================================%
\begin{table}[ht]
\center
\caption{AXP daily adjusted closing prices for a selected time range.}
\label{tab:axp-close}
\bigskip
% Table generated by Excel2LaTeX from sheet 'table'
\begin{tabular}{lrrrrrrrrr}
\hline
      Date &    1/16/01 &    1/17/01 &    1/18/01 & {\bf 1/19/01} & {\bf 1/22/01} &    1/23/01 &    1/24/01 &    1/25/01 &    1/26/01 \\
\hline
 Adj Close &      35.34 &      36.16 &       36.3 & {\bf 35.57} & {\bf 32.92} &      34.02 &      33.83 &      33.92 &      34.11 \\
\hline
\end{tabular}
\end{table}

%=======================================================Remarks============================================================%
\section{Conclusion}
%==========================================================================================================================%
The financial market today generates a huge amount of data every second, necessitating the development of new models and methods to analyze and take advantage of them. In this paper, we developed an interval-valued GARCH model for analyzing return range processes. It can be viewed as an extension of the point-valued GARCH model that allows for interval-valued returns. Yet its most important contribution lies in its capability to integrate realized measures and return-based model mechanism to produce ``information-rich" estimation of the volatility. Inferences of our Int-GARCH model can be made based on the conditional least squares method we proposed, which was shown to have quite stable and reliable performances. Our empirical study of the Dow Jones stocks data demonstrated the advantage of Int-GARCH model: it enriches the low-frequency volatility model such as GARCH with the high-frequency information, without having to suffer from the microstructure noises. This is especially valuable during crash time, or in general when the market is unstable, because drastic intraday variation is hardly captured by the low-frequency return only. We compared the performance of Int-GARCH to those of GARCH and RV. It was further shown that Int-GARCH has improved model flexibility: the GARCH mechanism prevents it from being too sensitive to data like the model-free RV;  the role of the intraday range also makes it less rigid than GARCH itself.

%======================================================Acknowledgment====================================================%
%\section*{Acknowledgment}

%======================================================bibliography========================================================%

%=======================================================Appendix==========================================================%
\appendix

%====================================================Proofs of theorems=====================================================%
\section{Proofs}
%======================================================Theorem 1=========================================================%
\subsection{Proof of Theorem \ref{thm:mean}}
\begin{proof}
We notice, by recursive calculations, that
\begin{eqnarray*}
h_{t} & = & \mu+x_{t-1}h_{t-1}\\
 & = & \mu+x_{t-1}\left(\mu+x_{t-2}h_{t-2}\right)\\
 & = & \mu+\mu x_{t-1}+x_{t-1}x_{t-2}h_{t-2}\\
 & = & \cdots\\
 & = & \mu+\mu x_{t-1}+\mu x_{t-1}x_{t-2}+\cdots+\mu x_{t-1}x_{t-2}\cdots x_{t-N}\\
 &  & +x_{t-1}x_{t-2}\cdots x_{t-(N+1)}h_{t-(N+1)}\\
 & = & \mu\left(1+\sum_{i=1}^{N}\prod_{j=1}^{i}x_{t-j}\right)+h_{t-(N+1)}\prod_{j=1}^{N+1}x_{t-j},\qquad\forall N\in\mathbb{N}.
\end{eqnarray*}
Taking expectations on both sides, we get
\begin{eqnarray*}
Eh_{t} & = & E\left(\mu\left(1+\sum_{i=1}^{N}\prod_{j=1}^{i}x_{t-j}\right)+h_{t-(N+1)}\prod_{j=1}^{N+1}x_{t-j}\right)\\
 & = & \mu\sum_{i=0}^{N}\left(Ex_{t}\right)^{i}+\left(Ex_{t}\right)^{N+1}Eh_{t-(N+1)},
\end{eqnarray*}
for all $N\in\mathbb{N}$. Letting $N\to\infty$, 
\begin{eqnarray*}
Eh_{t} & = & \mu\sum_{i=0}^{\infty}\left(Ex_{t}\right)^{i}+\lim_{N\rightarrow\infty}\left(Ex_{t}\right)^{N+1}Eh_{t-(N+1)}\\
 & = & \mu\sum_{i=0}^{\infty}\left(Ex_{t}\right)^{i}\\
 & = &  \dfrac{\mu}{1-E\left(x_{t}\right)}\\
 & = & \dfrac{\mu}{1-\alpha_{1}\sqrt{2/\pi}-\beta_{1}k-\gamma_{1}},
\end{eqnarray*}
since $|Ex_t|<1$. On the other hand, if $Ex_t\geq 1$, then $Eh_t>\mu\sum_{i=0}^{\infty}(Ex_t)^i=\infty$. Therefore, $Eh_t<\infty$ if and only if
\begin{eqnarray*}
\left|Ex_{t}\right| & = & \left|E\left(\alpha_{1}\left|\varepsilon_{t-1}\right|+\beta_{1}\eta_{t-1}+\gamma_{1}\right)\right|\\
 & = & \left|\alpha_{1}\sqrt{2/\pi}+\beta_{1}k+\gamma_{1}\right|\\
 & < & 1,
\end{eqnarray*}
and when this is satisfied, the Aumann expectation of $r_{t}$ is found to be
\begin{eqnarray*}
Er_{t} 
 & = & E\left[h_{t}\left(\varepsilon_{t}-\eta_{t}\right),h_{t}\left(\varepsilon_{t}+\eta_{t}\right)\right]\\
 & = & \left[E\left(h_{t}\right)E\left(\varepsilon_{t}-\eta_{t}\right), E\left(h_{t}\right)E\left(\varepsilon_{t}+\eta_{t}\right)\right]\\
 & = & \left[-kE\left(h_{t}\right),kE\left(h_{t}\right)\right].
\end{eqnarray*}
\end{proof}

%=====================================================Theorem 2=========================================================%
\subsection{Proof of Theorem \ref{thm:var}}
\begin{proof}
Recall, from the proof of Theorem \ref{thm:mean}, that 
$$h_{t}=\mu\left(1+\sum_{i=1}^{N}\prod_{j=1}^{i}x_{t-j}\right)+h_{t-(N+1)}\prod_{j=1}^{N+1}x_{t-j},\ \forall N\in\mathbb{N}.$$
Consequently, 
\begin{eqnarray*}
h_{t}^{2} & = & \left(\mu\left(1+\sum_{i=1}^{N}\prod_{j=1}^{i}x_{t-j}\right)+h_{t-(N+1)}\prod_{j=1}^{N+1}x_{t-j}\right)^{2}\\
 & = & \mu^{2}\left(1+\sum_{i=1}^{N}\prod_{j=1}^{i}x_{t-j}\right)^{2}+2\mu\left(1+\sum_{i=1}^{N}\prod_{j=1}^{i}x_{t-j}\right)\left(h_{t-(N+1)}\prod_{j=1}^{N+1}x_{t-j}\right)\\
 &  & +\left(h_{t-(N+1)}\prod_{j=1}^{N+1}x_{t-j}\right)^{2}\\
 & = & \mu^{2}\left(1+\sum_{i=1}^{N}\prod_{j=1}^{i}x_{t-j}\right)^{2}+2\mu h_{t-(N+1)}\prod_{j=1}^{N+1}x_{t-j}\\
 &  & +2\mu h_{t-(N+1)}\prod_{j=1}^{N+1}x_{t-j}\left(\sum_{i=1}^{N}\prod_{j=1}^{i}x_{t-j}\right)+h_{t-(N+1)}^{2}\prod_{j=1}^{N+1}x_{t-j}^{2},
 \forall N\in\mathbb{N}.
\end{eqnarray*}
Note that
\begin{eqnarray*}
 &  & \left(1+\sum_{i=1}^{N}\prod_{j=1}^{i}x_{t-j}\right)^{2}\\
 & = & \left(1+x_{t-1}+x_{t-1}x_{t-2}+\cdots+x_{t-1}x_{t-2}\cdots x_{t-N}\right)^{2}\\
 & = & 1+\sum_{i=1}^{N}\prod_{j=1}^{i}x_{t-j}^{2}+2\sum_{i=1}^{N}\prod_{j=1}^{i}x_{t-j}+2\sum_{i=1}^{N-1}\left[\left(\prod_{j=1}^{i}x_{t-j}^{2}\right)\left(\sum_{k=i+1}^{N}\,\prod_{l=i+1}^{k}x_{t-l}\right)\right],
\end{eqnarray*}
and
\begin{equation*}
  \,\,\,\,\prod_{j=1}^{N+1}x_{t-j}\left(\sum_{i=1}^{N}\prod_{j=1}^{i}x_{t-j}\right)
  =\sum_{i=1}^{N}\left[\left(\prod_{j=1}^{i}x_{t-j}^{2}\right)\left(\prod_{k=i+1}^{N+1}x_{t-k}\right)\right].
\end{equation*}
Therefore, 
\begin{eqnarray*}
h_{t}^{2} & = & \mu^{2}\left\{1+\sum_{i=1}^{N}\left(\prod_{j=1}^{i}x_{t-j}^{2}+2\prod_{j=1}^{i}x_{t-j}\right)+2\sum_{i=1}^{N-1}\left[\left(\prod_{j=1}^{i}x_{t-j}^{2}\right)\left(\sum_{k=i+1}^{N}\,\prod_{l=i+1}^{k}x_{t-l}\right)\right]\right\}\\
 &  & +2\mu h_{t-(N+1)}\prod_{j=1}^{N+1}x_{t-j}+2\mu h_{t-(N+1)}\sum_{i=1}^{N}\left[\left(\prod_{j=1`}^{i}x_{t-j}^{2}\right)\left(\prod_{k=i+1}^{N+1}x_{t-k}\right)\right]\\
 &  & +h_{t-(N+1)}^{2}\prod_{j=1}^{N+1}x_{t-j}^{2}\\
 & = & \mu^{2}\left\{1+\sum_{i=1}^{N}\left(\prod_{j=1}^{i}x_{t-j}^{2}+2\prod_{j=1}^{i}x_{t-j}\right)+2\sum_{i=1}^{N-1}\left[\left(\prod_{j=1}^{i}x_{t-j}^{2}\right)\left(\sum_{k=i+1}^{N}\,\prod_{l=i+1}^{k}x_{t-l}\right)\right]\right\}\\
 &  & +2\mu h_{t-(N+1)}\left(\prod_{j=1}^{N+1}x_{t-j}+\sum_{i=1}^{N}\left[\left(\prod_{j=1`}^{i}x_{t-j}^{2}\right)\left(\prod_{k=i+1}^{N+1}x_{t-k}\right)\right]\right)+h_{t-(N+1)}^{2}\prod_{j=1}^{N+1}x_{t-j}^{2},\\
 & & \forall N\in\mathbb{N}.
\end{eqnarray*}
Let $C_{1}=E\left(x_{t}\right)$ and $C_{2}=E\left(x_{t}^{2}\right)$.
%By Propositions~\ref{expected value of x_t}~and~\ref{expected value of x_t ^2}, we know that $E\left(x_{t}\right)$ and $E\left(x_{t}^{2}\right)$ both exist, so $C_{1}$ and $C_{2}$ are finite values. Notice that $h_{t-(N+1)}$ is only dependent on $x_{1},\ldots,x_{t-(N+2)}$ and hence, 
By the independence of $h_{t-(N+1)}$ and $x_{t-j}$, $\forall j\leq N+1$ and the fact that $\left\{ x_{t}\right\} $ are iid, we obtain,
\begin{eqnarray*}
E\left(h_{t}^{2}\right) & = & E\left\{\mu^{2}\left(1+\sum_{i=1}^{N}\left(\prod_{j=1}^{i}x_{t-j}^{2}+2\prod_{j=1}^{i}x_{t-j}\right)+2\sum_{i=1}^{N-1}\left[\left(\prod_{j=1}^{i}x_{t-j}^{2}\right)\left(\sum_{k=i+1}^{N}\,\prod_{l=i+1}^{k}x_{t-l}\right)\right]\right)\right.\\
 &  & \left.+2\mu h_{t-(N+1)}\left(\prod_{j=1}^{N+1}x_{t-j}+\sum_{i=1}^{N}\left[\left(\prod_{j=1`}^{i}x_{t-j}^{2}\right)\left(\prod_{k=i+1}^{N+1}x_{t-k}\right)\right]\right)+h_{t-(N+1)}^{2}\prod_{j=1}^{N+1}x_{t-j}^{2}\right\}\\
 & = & \mu^{2}+\mu^{2}\sum_{i=1}^{N}E\left[\prod_{j=1}^{i}x_{t-j}^{2}\right]+2\mu^{2}\sum_{i=1}^{N}E\left[\prod_{j=1}^{i}x_{t-j}\right]\\
 &  & +2\mu^{2}\sum_{i=1}^{N-1}E\left[\left(\prod_{j=1}^{i}x_{t-j}^{2}\right)\left(\sum_{k=i+1}^{N}\,\prod_{l=i+1}^{k}x_{t-l}\right)\right]\\
 &  & +2\mu E\left[h_{t-(N+1)}\right]\cdot\left(E\left[\prod_{j=1}^{N+1}x_{t-j}\right]+\sum_{i=1}^{N}E\left[\left(\prod_{j=1}^{i}x_{t-j}^{2}\right)\left(\prod_{k=i+1}^{N+1}x_{t-k}\right)\right]\right)\\
 &  & +E\left[h_{t-(N+1)}^{2}\right]\cdot E\left[\prod_{j=1}^{N+1}x_{t-j}^{2}\right]\\
 & = & \mu^{2}+\mu^{2}\sum_{i=1}^{N}C_{2}^{i}+2\mu^{2}\sum_{i=1}^{N}C_{1}^{i}+2\mu^{2}\sum_{i=1}^{N-1}\left[C_{2}^{i}\left(\sum_{k=i+1}^{N}C_{1}^{k-i}\right)\right]\\
 &  & +2\mu E\left[h_{t-(N+1)}\right]\cdot\left(C_{1}^{N+1}+\sum_{i=1}^{N}C_{2}^{i}C_{1}^{N-i+1}\right)+E\left[h_{t-(N+1)}^{2}\right]\cdot C_{2}^{N+1},\\
 &  & \forall N\in\mathbb{N}.
\end{eqnarray*}
The geometric series $\sum_{i=1}^{N}C_{2}^{i}$ and $\sum_{i=1}^{N}C_{1}^{i}$ converge if and only if $\left|C_{1}\right|<1$ and $\left|C_{2}\right|<1$. For the fourth term in the above expression,
\begin{eqnarray*}
\sum_{i=1}^{N-1}\left[C_{2}^{i}\left(\sum_{k=i+1}^{N}C_{1}^{k-i}\right)\right] %& = & \sum_{i=1}^{N-1}\left[C_{2}^{i}\left(C_{1}^{1}+C_{1}^{2}+\cdots+C_{1}^{N-i}\right)\right]\\
 & = & \sum_{i=1}^{N-1}\left[C_{2}^{i}\left(\dfrac{C_{1}\left(C_{1}^{N-i}-1\right)}{C_{1}-1}\right)\right]\\
 %& = & \dfrac{C_{1}}{C_{1}-1}\sum_{i=1}^{N-1}\left[C_{2}^{i}\left(C_{1}^{N-i}-1\right)\right]\\
 & = & \dfrac{C_{1}}{C_{1}-1}\sum_{i=1}^{N-1}\left[C_{2}^{i}C_{1}^{N-i}-C_{2}^{i}\right]\\
 & = & \dfrac{C_{1}}{C_{1}-1}\left[C_{1}^{N}\cdot\sum_{i=1}^{N-1}\left(\dfrac{C_{2}}{C_{1}}\right)^{i}-\sum_{i=1}^{N-1}C_{2}^{i}\right],
 \forall N\in\mathbb{N}.\\
\end{eqnarray*}
This is a finite number when $N\to\infty$ if and only if $|C_2|<|C_1|$. Finally,
\begin{eqnarray*}
\sum_{i=1}^{N}C_{2}^{i}C_{1}^{N-i+1} & = & C_{1}^{N+1}\sum_{i=1}^{N}\left(\dfrac{C_{2}}{C_{1}}\right)^{i},
\forall N\in\mathbb{N}
\end{eqnarray*}
is finite for $N\to\infty$ if and only if $|C_{2}|<|C_{1}|$. Therefore, $E\left(h_t^2\right)<\infty$ if and only if $|C_2|<|C_1|<1$, or $C_2<C_1<1$, since both $C_1$ and $C_2$ are positive.\\ 
\indent Under this assumption and in addition that $Eh^2_{-\infty}<\infty$, letting $N\to\infty$, we find the second moment of $h_t^2$ to be
\begin{eqnarray*}
E\left(h_{t}^{2}\right) 
 & = & \mu^{2}\left(1-\dfrac{C_{2}}{C_{2}-1}-2\dfrac{C_{1}}{C_{1}-1}+2C_{1}C_{2}\dfrac{1}{\left(C_{2}-1\right)\left(C_{1}-1\right)}\right)\\
 & = & \mu^{2}\left(\dfrac{C_{1}C_{2}-C_{1}-C_{2}+1-C_{2}C_{1}+C_{2}-2C_{1}C_{2}+2C_{1}+2C_{1}C_{2}}{\left(C_{2}-1\right)\left(C_{1}-1\right)}\right)\\
 & = & \mu^{2}\dfrac{C_{1}+1}{\left(C_{2}-1\right)\left(C_{1}-1\right)}.
\end{eqnarray*}
Consequently, the unconditional variance of $r_t$ is
\begin{eqnarray*}
\mbox{Var}(r_{t})  & = & \mbox{Var}\left(h_{t}\varepsilon_{t}\right)+\mbox{Var}\left(h_{t}\eta_{t}\right)\\
 & = & \left(1+k+k^{2}\right)E\left(h_{t}^{2}\right)-k^{2}\left[E\left(h_{t}\right)\right]^{2}.
\end{eqnarray*}
This completes the proof.
\end{proof}

%=====================================================Theorem 3=========================================================%
\subsection{Proof of Theorem \ref{thm:cov}}
\begin{proof}
First we notice, $\forall t, s\in\mathbb{N}$,
\begin{eqnarray*}
  r_{t} &= &\left[h_{t}\varepsilon_{t}-h_{t}\eta_{t},h_{t}\varepsilon_{t}+h_{t}\eta_{t}\right],\\
  r_{t+s} &=&\left[h_{t+s}\varepsilon_{t+s}-h_{t+s}\eta_{t+s},h_{t+s}\varepsilon_{t+s}+h_{t+s}\eta_{t+s}\right],
\end{eqnarray*}
and therefore,
\begin{equation}\label{eqn1}
\mbox{Cov}\left(r_{t},r_{t+s}\right)=\mbox{Cov}\left(h_{t}\varepsilon_{t},h_{t+s}\varepsilon_{t+s}\right)+\mbox{Cov}\left(h_{t}\eta_{t},h_{t+s}\eta_{t+s}\right).
\end{equation}
The first term
\begin{eqnarray}
\mbox{Cov}\left(h_{t}\varepsilon_{t},h_{t+s}\varepsilon_{t+s}\right) 
& = & E\left[\left(h_{t}\varepsilon_{t}-E\left(h_{t}\epsilon_{t}\right)\right)\left(h_{t+s}\varepsilon_{t+s}-E\left(h_{t+s}\epsilon_{t+s}\right)\right)\right]\nonumber\\
 & = & E\left(h_{t}h_{t+s}\varepsilon_{t}\cdot\varepsilon_{t+s}\right)\nonumber\\
 & = & \begin{cases}
E\left(h_{t}^{2}\varepsilon_{t}^{2}\right),\,\, & s=0\\
E\left(h_{t}h_{t+s}\varepsilon_{t}\right)\cdot E\left(\varepsilon_{t+s}\right),\,\, & |s|>0
\end{cases}\nonumber\\
 & = & \begin{cases}
E\left(\varepsilon_{t}^{2}\right)\cdot E\left(h_{t}^{2}\right),\,\, & s=0\\
0, & |s|>0
\end{cases}\nonumber\\
 & = & \begin{cases}\label{eqn2}
E\left(h_{t}^{2}\right),\,\, & s=0\\
0, & |s|>0
\end{cases}
\end{eqnarray}
since $\left\{ \varepsilon_{t}\right\} $ are i.i.d. \\

Similarly, the second term becomes
\begin{eqnarray}
 &  & \mbox{Cov}\left(h_{t}\eta_{t},h_{t+s}\eta_{t+s}\right)\nonumber\\
 & = & E\left[\left(h_{t}\eta_{t}-kE\left(h_{t}\right)\right)\left(h_{t+s}\eta_{t+s}-kE\left(h_{t+s}\right)\right)\right]\nonumber\\
 & = & E\left(h_{t}h_{t+s}\eta_{t}\eta_{t+s}\right)-kE\left(h_{t}\right)E\left(h_{t+s}\eta_{t+s}\right)-kE\left(h_{t+s}\right)E\left(h_{t}\eta_{t}\right)+k^{2}E\left(h_{t}\right)E\left(h_{t+s}\right)\nonumber\\
 & = & E\left(h_{t}h_{t+s}\eta_{t}\eta_{t+s}\right)-kE\left(h_{t}\right)kE\left(h_{t+s}\right)-kE\left(h_{t+s}\right)kE\left(h_{t}\right)+k^{2}E\left(h_{t}\right)E\left(h_{t+s}\right)
 \nonumber\\
 & = & E\left(h_{t}h_{t+s}\eta_{t}\eta_{t+s}\right)-k^{2}E\left(h_{t}\right)E\left(h_{t+s}\right)\nonumber\\
 & = & \begin{cases}
E\left(h_{t}^{2}\eta_{t}^{2}\right)-k^{2}\left[E\left(h_{t}\right)\right]^{2}, & s=0\\
E\left(h_{t}h_{t+s}\eta_{t}\right)E\left(\eta_{t+s}\right)-k^{2}E\left(h_{t}\right)E\left(h_{t+s}\right), & |s|>0
\end{cases}\nonumber\\
 & = & \begin{cases}\label{eqn3}
\left(k+k^{2}\right)E\left(h_{t}^{2}\right)-k^{2}\left[E\left(h_{t}\right)\right]^{2}, & s=0\\
kE\left(h_{t}h_{t+s}\eta_{t}\right)-k^{2}\left[E\left(h_{t}\right)\right]^2, & |s|>0.
\end{cases}
\end{eqnarray}
Plugging (\ref{eqn2}) and (\ref{eqn3}) into (\ref{eqn1}), we obtain
\begin{align*}
\mbox{Cov}\left(r_{t},r_{t+s}\right) 
 & =\begin{cases}
\left(1+k+k^{2}\right)E\left(h_{t}^{2}\right)-k^{2}\left[E\left(h_{t}\right)\right]^{2}, & s=0\\
kE\left(h_{t}h_{t+s}\eta_{t}\right)-k^{2}\left[E\left(h_{t}\right)\right]^{2}, & |s|>0, 
\end{cases}
\end{align*}
where $E\left(h_{t}\right)=\dfrac{\mu}{1-C_{1}}$, $E\left(h_{t}^{2}\right)=\mu^{2}\dfrac{C_{1}+1}{\left(C_{2}-1\right)\left(C_{1}-1\right)}$,
and 
\begin{equation*}
  E\left(h_{t}h_{t+s}\eta_{t}\right)
  =\dfrac{\mu^{2}k}{C_{1}-1}\left(-\dfrac{C_{1}^{s}-1}{C_{1}-1}+\dfrac{C_{1}^{s}+C_{1}^{s-1}}{C_{2}-1}\cdot\left[\alpha_{1}\sqrt{\dfrac{2}{\pi}}
  +  \beta_{1}\left(1+k\right)+\gamma_{1}\right]\right)
\end{equation*}
(see Lemma~\ref{eta-X and h-h-eta}, Appendix).
\end{proof}

%=====================================================Theorem 4=========================================================%
\subsection{Proof of Theorem \ref{thm:mean-gen}}
\begin{proof}
By the definition of (\ref{def:X_it}), (\ref{igarch_5}) can be rewritten as
\begin{eqnarray*}
h_t&=&\mu+\sum_{i=1}^{p}\alpha_i|\epsilon_{t-i}|h_{t-i}+\sum_{i=1}^{q}\beta_i\eta_{t-i}h_{t-i}+\sum_{i=1}^{w}\gamma_ih_{t-i}\\
&=&\mu+\sum_{i=1}^{k}x_{i,t-i}h_{t-i}.
\end{eqnarray*}
Expanding $h_t$ recursively, we obtain
\begin{eqnarray}
h_t&=&\mu+\sum_{i=1}^{k}x_{i,t-i}\left(\mu+\sum_{j=1}^{k}x_{j,t-i-j}h_{t-i-j}\right)\nonumber\\
&=&\mu\left(1+\sum_{i=1}^{k}x_{i,t-i}\right)+\sum_{i=1}^{k}\sum_{j=1}^{k}x_{i,t-i}x_{j,t-i-j}h_{t-i-j}\nonumber\\
&=&\cdots\nonumber\\
&=&\mu\left[1+\sum_{n=1}^{N}\sum_{i_1=1}^{k}\cdots\sum_{i_n=1}^{k}\left(\prod_{j=1}^{n}x_{i_j, t-i_1-\cdots-i_j}\right)\right]\nonumber\\
&&+\sum_{i_1=1}^{k}\cdots\sum_{i_n{N+1}=1}^{k}\left(\prod_{j=1}^{N+1}x_{i_j, t-i_1-\cdots-i_j}\right)h_{t-i_1-\cdots-i_{N+1}}.\label{exp-ht}
\end{eqnarray}
Notice that $x_{i,t}$ and $x_{j,s}$ are independent, $\forall i,j\in\mathbb{N}$ and $t\neq s$. Taking expectations on both sizes of (\ref{exp-ht}), we get
\begin{eqnarray*}
E\left(h_t\right)&=&\mu\left[1+\sum_{n=1}^{N}\sum_{i_1=1}^{k}\cdots\sum_{i_n=1}^{k}\left(\prod_{j=1}^{n}\mu_{i_j}\right)\right]\\
&&+\sum_{i_1=1}^{k}\cdots\sum_{i_{N+1}=1}^{k}\left(\prod_{j=1}^{N+1}\mu_{i_j}\right)E\left(h_{t-i_1-\cdots-i_{N+1}}\right)\\
&\leq&\mu\left[1+\sum_{n=1}^{N}\left(\sum_{j=1}^{k}\mu_{i_j}\right)^n\right]\\
&&+\sum_{i_1=1}^{k}\cdots\sum_{i_{N+1}=1}^{k}\left(\prod_{j=1}^{N+1}\mu_{i_j}\right)\max\left\{E\left(h_{t-l}\right): N+1\leq l\leq k(N+1)\right\}\\
&\leq&\mu\left[1+\sum_{n=1}^{N}\left(\sum_{j=1}^{k}\mu_{i_j}\right)^n\right]\\
&&+\max\left\{E\left(h_{t-l}\right): N+1\leq l\leq k(N+1)\right\}\left(\sum_{j=1}^{k}\mu_{i_j}\right)^{N+1},
\end{eqnarray*}
$\forall\ N\in\mathbb{N}$. Therefore, $E\left(h_t\right)\leq\infty$ if and only if 
\begin{equation*}
  \sum_{i=1}^{k}\mu_i<1.
\end{equation*}
When it is satisfied,
\begin{eqnarray*}
E\left(h_t\right)&=&\lim_{N\to\infty}\mu\left[1+\sum_{n=1}^{N}\left(\sum_{j=1}^{k}\mu_{i_j}\right)^n\right]\\
&&+\lim_{N\to\infty}\max\left\{E\left(h_{t-l}\right): N+1\leq l\leq k(N+1)\right\}\left(\sum_{j=1}^{k}\mu_{i_j}\right)^{N+1}\\
&=&\frac{\mu}{1-\sum_{i=1}^{k}\mu_i},
\end{eqnarray*}
by the finiteness of $E\left(h_{-\infty}\right)$. The formula for $E\left(r_t\right)$ follows immediately from the Aumann expectation.
\end{proof}

%=====================================================Corollary 1=========================================================%
\subsection{Proof of Corollary \ref{cor:stat}}
\begin{proof}
It is immediate from Theorem \ref{thm:mean}, \ref{thm:var}, and \ref{thm:cov}.
\end{proof}

%=====================================================Corollary 2=========================================================%
\subsection{Proof of Corollary \ref{cor:acf}}
\begin{proof}
The Auto-correlation Function (ACF) of $\left\{ r_{t}\right\} $ is
defined to be $\rho(s)=\mbox{Corr}\left(r_{t},r_{t+s}\right)=\dfrac{\gamma(s)}{\gamma(0)}$.
Then the ACF of $\left\{ r_{t}\right\} $ is 
\[
\rho(s)=\begin{cases}
1, & s=0\\
\dfrac{kE\left(h_{t}h_{t+s}\eta_{t}\right)-k^{2}\left[E\left(h_{t}\right)\right]^{2}}{\left(1+k+k^{2}\right)E\left(h_{t}^{2}\right)-k^{2}\left[E\left(h_{t}\right)\right]^{2}}, & |s|>0.
\end{cases}
\]
\end{proof}

%========================================================================================================================%
\section{Lemmas}

%======================================================Lemma 1============================================================%
\begin{lemma}\label{eta-X and h-h-eta}
\begin{eqnarray*}
  (i)\ &&E\left(\eta_{t}x_{t}\right)=\alpha_{1}\sqrt{\dfrac{2}{\pi}}\cdot k+\beta_{1}\left(k+k^{2}\right)+\gamma_{1}k;\\
 (ii)\ &&E\left(h_{t}h_{t+s}\eta_{t}\right)=\dfrac{\mu^{2}k}{C_{1}-1}\left(-\dfrac{C_{1}^{s}-1}{C_{1}-1}+\dfrac{C_{1}^{s}+C_{1}^{s-1}}{C_{2}-1}\cdot\left[\alpha_{1}\sqrt{\dfrac{2}{\pi}}+\beta_{1}\left(1+k\right)+\gamma_{1}\right]\right).
\end{eqnarray*}
 \end{lemma} 

\begin{proof} 

(i)
\begin{eqnarray*}
E\left(\eta_{t}x_{t}\right) & = & E\left[\eta_{t}\cdot\left(\alpha_{1}\left|\varepsilon_{t}\right|+\beta_{1}\eta_{t}+\gamma_{1}\right)\right]\\
 & = & \alpha_{1}E\left(\left|\varepsilon_{t}\right|\right)\cdot E\left(\eta_{t}\right)+\beta_{1}\cdot E\left(\eta_{t}^{2}\right)+\gamma_{1}E\left(\eta_{t}\right)\\
 & = & \alpha_{1}\sqrt{\dfrac{2}{\pi}}\cdot k+\beta_{1}\left(k+k^{2}\right)+\gamma_{1}k
\end{eqnarray*}

(ii) First, we expand $h_{t+s}$ recursively:
\begin{eqnarray*}
h_{t+s} & = & \mu+x_{t+s-1}h_{t+s-1}\\
 & = & \mu+\mu x_{t+s-1}+x_{t+s-1}x_{t+s-2}h_{t+s-2}\\
 & = & \cdots\\
 & = & \mu+\mu x_{t+s-1}+\mu x_{t+s-1}x_{t+s-2}+\cdots\\
 &  & +\mu x_{t+s-1}x_{t+s-2}\cdots x_{t+1}+x_{t+s-1}x_{t+s-2}\cdots x_{t}h_{t}\\
 & = & \mu\left(1+\sum_{i=1}^{s-1}\prod_{j=1}^{i}x_{t+s-j}\right)+h_{t}\prod_{j=1}^{s}x_{t+s-j}.
\end{eqnarray*}
Consequently,
\begin{eqnarray*}
h_{t}h_{t+s}\eta_{t} & = & h_{t}\eta_{t}\cdot\left(\mu\left(1+\sum_{i=1}^{s-1}\prod_{j=1}^{i}x_{t+s-j}\right)+h_{t}\prod_{j=1}^{s}x_{t+s-j}\right)\\
 & = & \mu h_{t}\eta_{t}+\mu h_{t}\eta_{t}\sum_{i=1}^{s-1}\prod_{j=1}^{i}x_{t+s-j}+h_{t}^{2}\eta_{t}\prod_{j=1}^{s}x_{t+s-j}.
\end{eqnarray*}

Then the expected value is found to be
\begin{eqnarray*}
E\left(h_{t}h_{t+s}\eta_{t}\right) & = & \mu\cdot E\left(h_{t}\right)\cdot E\left(\eta_{t}\right)+\mu\cdot E\left(h_{t}\right)\cdot E\left(\eta_{t}\right)\cdot E\left[\sum_{i=1}^{s-1}\prod_{j=1}^{i}\left(x_{t+s-j}\right)\right]\\
 &  & +E\left(h_{t}^{2}\right)\cdot E\left(\eta_{t}\prod_{j=1}^{s}x_{t+s-j}\right)\\
 & = & \mu k\cdot E\left(h_{t}\right)+\mu k\cdot E\left(h_{t}\right)\cdot\left[\sum_{i=1}^{s-1}\prod_{j=1}^{i}E\left(x_{t+s-j}\right)\right]\\
 &  & +E\left(h_{t}^{2}\right)\cdot E\left(\eta_{t}x_{t}\right)\cdot\prod_{j=1}^{s-1}E\left(x_{t+s-j}\right)\\
 & = & \mu k\cdot E\left(h_{t}\right)+\mu k\cdot E\left(h_{t}\right)\cdot\dfrac{C_{1}\left(C_{1}^{s-1}-1\right)}{C_{1}-1}\\
 &  & +E\left(h_{t}^{2}\right)\cdot E\left(\eta_{t}x_{t}\right)\cdot C_{1}^{s-1},
\end{eqnarray*}
where $C_{1}=E\left(x_{t}\right)$.\\

Finally, remembering that $E\left(h_{t}^{2}\right)=\mu^{2}\dfrac{C_{1}+1}{\left(C_{2}-1\right)\left(C_{1}-1\right)}$, the above calculation is simplified to
\begin{eqnarray*}
E\left(h_{t}h_{t+s}\eta_{t}\right) & = & \mu k\cdot\dfrac{\mu}{1-C_{1}}\cdot\left[1+\dfrac{C_{1}\left(C_{1}^{s-1}-1\right)}{C_{1}-1}\right]\\
 &  & +E\left(h_{t}^{2}\right)\cdot\left[\alpha_{1}\sqrt{\dfrac{2}{\pi}}\cdot k+\beta_{1}\left(k+k^{2}\right)+\gamma_{1}k\right]\cdot C_{1}^{s-1}\\
 & = & \mu^{2}k\cdot\dfrac{1}{1-C_{1}}\cdot\left[\dfrac{C_{1}-1+C_{1}\left(C_{1}^{s-1}-1\right)}{C_{1}-1}\right]\\
 &  & +E\left(h_{t}^{2}\right)\cdot\left[\alpha_{1}\sqrt{\dfrac{2}{\pi}}\cdot k+\beta_{1}\left(k+k^{2}\right)+\gamma_{1}k\right]\cdot C_{1}^{s-1}\\
 & = & -\mu^{2}k\cdot\dfrac{C_{1}^{s}-1}{\left(C_{1}-1\right)^{2}}+\mu^{2}\dfrac{C_{1}+1}{\left(C_{2}-1\right)\left(C_{1}-1\right)}\cdot\left[\alpha_{1}\sqrt{\dfrac{2}{\pi}}\cdot k+\beta_{1}\left(k+k^{2}\right)+\gamma_{1}k\right]\cdot C_{1}^{s-1}\\
 & = & \dfrac{\mu^{2}k}{C_{1}-1}\left(-\dfrac{C_{1}^{s}-1}{C_{1}-1}+\dfrac{C_{1}^{s}+C_{1}^{s-1}}{C_{2}-1}\cdot\left[\alpha_{1}\sqrt{\dfrac{2}{\pi}}+\beta_{1}\left(1+k\right)+\gamma_{1}\right]\right)
\end{eqnarray*}
\end{proof}

%======================================================Lemma 2============================================================%
\begin{lemma}\label{expected value of x_t ^2}
\begin{eqnarray*}
E\left(x_{t}^{2}\right) & = & \alpha_{1}^{2}+\beta_{1}^{2}\left(k+k^{2}\right)+\gamma_{1}^{2}+2\alpha_{1}\beta_{1}\sqrt{\dfrac{2}{\pi}}k+2\alpha_{1}\gamma_{1}\sqrt{\dfrac{2}{\pi}}+2\beta_{1}\gamma_{1}k.
\end{eqnarray*}
\end{lemma}
\begin{proof}
\begin{eqnarray*}
E\left(x_{t}^{2}\right) & = & E\left(\left(\alpha_{1}\left|\varepsilon_{t}\right|+\beta_{1}\eta_{t}+\gamma_{1}\right)^{2}\right)\\
 & = & E\left(\alpha_{1}^{2}\varepsilon_{t}^{2}+\beta_{1}^{2}\eta_{t}^{2}+\gamma_{1}^{2}+2\alpha_{1}\beta_{1}\left|\varepsilon_{t}\right|\eta_{t}+2\alpha_{1}\gamma_{1}\left|\varepsilon_{t}\right|+2\beta_{1}\gamma_{1}\eta_{t}\right)\\
 & = & \alpha_{1}^{2}E\left(\varepsilon_{t}^{2}\right)+\beta_{1}^{2}E\left(\eta_{t}^{2}\right)+\gamma_{1}^{2}+2\alpha_{1}\beta_{1}E\left(\left|\varepsilon_{t}\right|\right)\cdot E\left(\eta_{t}\right)+2\alpha_{1}\gamma_{1}E\left(\left|\varepsilon_{t}\right|\right)+2\beta_{1}\gamma_{1}E\left(\eta_{t}\right)\\
 %&  & \mbox{since we know that \ensuremath{\varepsilon_{t}}and \ensuremath{\eta_{t}}\,\ are iid}\\
 & = & \alpha_{1}^{2}+\beta_{1}^{2}\left(k+k^{2}\right)+\gamma_{1}^{2}+2\alpha_{1}\beta_{1}\sqrt{\dfrac{2}{\pi}}k+2\alpha_{1}\gamma_{1}\sqrt{\dfrac{2}{\pi}}+2\beta_{1}\gamma_{1}k.
\end{eqnarray*}
\end{proof}

%================================================================================

\end{document}